\documentclass[11pt]{article}
\usepackage{amsfonts,amsmath,amssymb}
\usepackage{xcolor}
\usepackage{graphicx}
\usepackage{rotating}
\usepackage{multicol}
\usepackage{float}
\usepackage{multirow}
\usepackage{enumerate}
\usepackage{wasysym}
\usepackage{subfigure}
\usepackage[normalem]{ulem}
\usepackage[flushleft]{threeparttable} 
\usepackage[authoryear,round]{natbib}
\usepackage[colorlinks, citecolor=blue, linkcolor=blue]{hyperref}
\usepackage{xr}
\usepackage{bbm}
\usepackage{todonotes}

\DeclareMathOperator*{\E}{\mathbb{E}} 
\DeclareMathOperator*{\1}{\mathbbm{1}} 

\allowdisplaybreaks
\numberwithin{equation}{section}
\usepackage{ntheorem}
\theoremseparator{:}

\newtheorem{thm}{Theorem}
\newtheorem{cl}{Corollary}
\newtheorem{lm}{Lemma}
\newtheorem{df}{Definition}
\newtheorem{pp}{Proposition}

\newtheorem{as}{Assumption}[section]

\allowdisplaybreaks

\newcommand{\blue}{\textcolor{blue}}

\newcommand{\qed}{\hfill \ensuremath{\Box}}

\begin{document}
\author{Nan Liu\footnote{Nan Liu, Assistant Professor, Paula and Gregory Chow Institute for Studies in Economics, Wang Yanan Institute for Studies in Economics (WISE), Department of Statistics \& Data Science, School of Economics, Xiamen University.}
\hspace{1cm}
Yanbo Liu\footnote{Yanbo Liu, Associate Professor, School of Economics, Shandong University.}
\hspace{1cm}
Yuya Sasaki\footnote{Yuya Sasaki, Brian and Charlotte Grove Chair and Professor of Economics, Department of Economics, Vanderbilt University.}
\hspace{1cm}
Yuanyuan Wan\footnote{Yuanyuan Wan, Associate Professor, Department of Economics, University of Toronto.}}
\title{{\Large \textbf{Root-$n$ Asymptotically Normal Maximum Score Estimation\thanks{Sasaki gratefully acknowledges research support from Brian and Charlotte Grove.}}}}

\maketitle

\begin{abstract}
The maximum score method \citep{manski1975maximum,manski1985semiparametric} is a powerful approach for binary choice models, yet it is known to face both practical and theoretical challenges. In particular, the estimator converges at a slower-than-root-$n$ rate to a nonstandard limiting distribution. 
We investigate conditions under which strictly concave surrogate score functions can be employed to achieve identification through a smooth criterion function. 
This criterion enables root-$n$ convergence to a normal limiting distribution.
While the conditions to guarantee these desired properties are nontrivial, we characterize them in terms of primitive conditions.
Extensive simulation studies support, the root-$n$ convergence rate, the asymptotic normality, and the validity of the standard inference methods.
\vskip0.4cm
	
	
\vskip1cm
	
\baselineskip=15pt
\end{abstract}

\newpage
%%%%%%%%%%%%%%%%%%%%%%%%%%%%%%%%%%%%%%%%%
\section{Introduction}\label{sec:introduction}
%%%%%%%%%%%%%%%%%%%%%%%%%%%%%%%%%%%%%%%%%

The maximum score method \citep*{manski1975maximum,manski1985semiparametric} is a powerful approach for identifying and estimating the parameters of binary choice models without imposing distributional assumptions on the error term. 
At the same time, it is well known to face several challenges in both theory and practice.  In particular, \citet*{kim1990cube} show that the estimated parameters for the maximum score estimands (and other related estimands) converge at the cube-root-$n$ rate to a nonstandard (non-Gaussian) limiting distribution. 
This behavior arises from the discontinuous sample criterion involving the indicator function.
Moreover, the na\"ive bootstrap is not valid in this setting \citep*[see][]{abrevaya2005bootstrap, leger2006bootstrap}.

The literature has developed several methods for estimation and inference in this context, including the smoothing approach \citep*{horowitz1992smoothed}, the subsampling approach \citep*{delgado2001subsampling, seo2018local}, the $m$-out-of-$n$ bootstrap approach \citep*{lee2006m, bickel2011resampling}, modified bootstrap procedures \citep*{cattaneo2020bootstrap, cattaneo2024bootstrap, cheng2024inference},\footnote{See also \citet{cattaneo2026continuity} for the validity of confidence intervals constructed by the modified bootstrap.} and finite-sample distribution \citep*{rosen2025}, among others, designed to accommodate or address the nonstandard limit property.

The statistical learning literature proposes the use of surrogate loss functions in place of the indicator function \citep*[e.g.,][]{lugosi2004bayes, zhang2004statistical, steinwart2005consistency, bartlett2006convexity, zhao2012estimating} to enable convex optimization. 
This idea is well suited to the maximum score and related methods and has also been adopted in econometric work in related contexts.

\citet*{babii2020binary} introduce surrogate loss functions for binary classification in rich or nonparametric function classes and analyze the resulting trade-offs between classification risk and complexity.
\citet*{kitagawa2023constrained} study surrogate losses in constrained classification settings that may involve misspecification of the classifier class, and show that the hinge loss is essentially the only surrogate that preserves consistency in such environments.
\citet*{chen2025relu} employ the ReLU surrogate, propose a two-step estimation procedure, and derive asymptotic normality at the rate $n^{s/(2s+1)}$, where $s$ denotes the smoothness order of the conditional expectation function of the binary choice $Y$ given $X$.
\citet*{liu2025nonparametric} consider a nonparametric class of classifiers and advocate strictly convex surrogate functions to achieve point identification of a nonparametric representative, from which they derive standard nonparametric rates of convergence.

This paper focuses on a parametric family of classifiers and establishes conditions under which strictly convex surrogate score functions can point identify a representative classifying parameter vector, so that the resulting one-step estimator achieves root-$n$ asymptotic normality without the need of nonparametric nuisance estimation. 
This objective differs from those of the papers discussed in the previous paragraph. 

Indeed, like \citet*{liu2025nonparametric}, we advocate strictly convex surrogate score functions, which rule out the hinge and ReLU losses in particular. 
However, establishing the validity of such surrogate methods is more challenging in the parametric framework:
whereas the pointwise minimum-risk preservation property of surrogate losses as in \citet*{bartlett2006convexity} is directly applicable in the nonparametric setting, that pointwise argument does not apply under parametric restrictions. 
To validate the surrogate approach we need to impose restrictions on the class of distributions of $X$. 
Characterizing such conditions is the central contribution of this paper.

Under these conditions, the parameters can be estimated via a single-step maximization of the strictly concave and smooth surrogate score, without requiring nonparametric nuisance estimation. 
Consequently, the root-$n$ rate of convergence can be established using standard arguments. 
Note that the availability of these ``standard arguments'' stands in contrast to the maximum score literature, where the asymptotic properties are known to be nonstandard.
Extensive simulation studies further demonstrate that, under these conditions, our proposed method achieves root-$n$ convergence to a normal limiting distribution.

Indeed, \citet*{horowitz1992smoothed} notes that the convergence rate cannot be faster than the nonparametric rate $n^{-s/(2s+1)}$ in the minimax sense, where $s$ denotes the smoothness order of the conditional distribution function of the error $\varepsilon$ given $X$.  Root-n consistency, however, can be achieved under additional assumptions, such as with the single index specification \citep[e.g.,][]{klein1993efficient}. In this paper, we provide a novel identification approach that builds on conditions that are comparable to those in \citet{klein1993efficient}. Specifically, we show that these restrictions, stated as conditions (T.\ref{thm:main}.1) and (T.\ref{thm:main}.2) in Theorem \ref{thm:main}, ensure that the smooth surrogate maximum score identifies the original parameters, thereby yielding root-$n$ consistency and asymptotic normality. Furthermore, our estimation procedure does not involve any infinite dimensional nuisance parameters. Computing our estimator only requires maximizing a smooth objective function over a finite-dimensional parameter space, with a unique maximizer and no need for trimming or tuning.

Thus, our claim does not contradict the existing knowledge in the literature.
We view the alternative approaches in the literature discussed above as complementary to our work. 
They impose different restrictions on the data-generating processes, which in turn lead to different identification strategies, estimation methods, asymptotic properties, and/or research objectives.

In summary, we investigate conditions under which a strictly concave surrogate maximum score identifies a solution to the original maximum score problem, thereby enabling root-$n$ asymptotic normality of the sample counterpart one-step estimator. 
Consequently, standard inference methods (such as the bootstrap) are valid, and there is no need to select tuning parameters for subsampling or smoothing, unlike in conventional maximum score methods.
An important advantage of the availability of standard asymptotic results is that they facilitate implementation using widely used statistical software packages, such as Stata, whose default output assumes a normal limiting distribution.

\bigskip\noindent
{\bf Organization:}
The rest of the paper is organized as follows. 
Section \ref{sec:setup} introduces the model setup. 
Section \ref{sec:main} presents the main results. 
Section \ref{sec:primitive} discusses primitive conditions for our key restrictions and provides illustrative examples. 
Section \ref{sec:estimation} introduces the sample analog estimator and establishes its asymptotic properties. 
Section \ref{sec:simulation} presents extensive simulation studies. 
Section \ref{sec:conclusion} concludes. 
All mathematical proofs are provided in the appendix.

%%%%%%%%%%%%%%%%%%%%%%%%%%%%%%%%%%%%%%%%%
\section{Maximum Score and Its Surrogate Counterpart}\label{sec:setup}
%%%%%%%%%%%%%%%%%%%%%%%%%%%%%%%%%%%%%%%%%

Consider the threshold-crossing binary choice model
\begin{equation}\label{eq:threshold_crossing}
Y = \1\{ X'b_0 + \varepsilon \ge 0\}
\end{equation}
where $Y$ is a $\{0,1\}$-valued choice variable, $X$ is a $\mathbb{R}^d$-valued explanatory variables, $b_0$ is an unknown parameter vector, and $\varepsilon$ is an unobserved error satisfying the conditional median restriction
\begin{equation}\label{eq:conditional_median}
\mathrm{Median}(\varepsilon|X) = 0.
\end{equation}
For this model, Manski develops the maximum score method to identify and estimate $b_0$.

The maximum score method solves 
$\max_{b \in B} Q_0(b)$, where $Q_0$ is defined by
\begin{align}\label{eq:q0}
Q_0(b) 
= \mathbb{E}\big[ Y \cdot \1\{X'b \ge 0\} 
+ (1-Y) \cdot \1\{X'b < 0\} \big].
\end{align}
When this problem is replaced by its sample counterpart, both practical and theoretical difficulties arise. 
Because the score is defined through an indicator function, which is discontinuous, the resulting optimization problem is non-convex. 
Moreover, the resulting estimator converges at a nonstandard rate that is slower than $\sqrt{n}$ and has a nonstandard limiting distribution.

To circumvent these difficulties, the statistical learning literature proposes replacing the indicator function with a continuous surrogate loss. 
Let $\phi:\mathbb{R}\to\mathbb{R}$ be a measurable surrogate function. 
Define the surrogate objective function $Q_\phi$ by
\begin{align}\label{eq:q_phi}
Q_\phi(b)
= \mathbb{E}\big[ Y \cdot \phi(X'b) + (1-Y) \cdot \phi(-X'b) \big].
\end{align}
We then define the surrogate maximum score method as the solution to
$
\max_{b \in B} Q_\phi(b).
$

There is, however, a natural question at this point. 
Are the solutions to the surrogate problem guaranteed to also solve the original maximum score problem? 
If the answer to this question is `no,' then the surrogate method would generally yield biased estimates and would therefore be of limited usefulness. 
On the other hand, if the answer is `yes,' then the surrogate approach could indeed help overcome the difficulties associated with the maximum score method.

We argue that, under certain conditions, the surrogate maximum score method indeed yields solutions to the original maximum score problem -- Section \ref{sec:validity}. 
Furthermore, the surrogate problem allows for point identification without imposing additional restrictions -- Section \ref{sec:existence_uniqueness}. 
The key conditions under which these desirable results hold are nontrivial. 
Nevertheless, they accommodate a wide class of distributions for $X$ -- Section \ref{sec:primitive}.

%%%%%%%%%%%%%%%%%%%%%%%%%%%%%%%%%%%%%%%%%
\section{Main Results}\label{sec:main}
%%%%%%%%%%%%%%%%%%%%%%%%%%%%%%%%%%%%%%%%%
\subsection{Validity of the Surrogate Maximum Score}\label{sec:validity}
%%%%%%%%%%%%%%%%%%%%%%%%%%%%%%%%%%%%%%%%%

Denoting the conditional choice probability by
$$
\eta(x) = \Pr(Y=1|X=x),
$$
consider the following assumption on the Bayes boundary.

%%%%%%%%%%%%%%%%%%%%%%%%%%%%%%%%%%%%%%%%%
\begin{as}[Bayes Boundary]\label{as:bayes}
Let $B \subseteq \mathbb{R}^d$ be nonempty and compact.
There exists $b_0 \in B$ such that the following conditions hold.
\begin{enumerate}[(i)]
\item\label{as:bayes:correct} $\1\{\eta(X) \ge 1/2\} = \1\{X'b_0 \ge 0\}$ a.s.
\item\label{as:bayes:nondegeneracy} $\Pr(X'b_0 > 0)>0$ and $\Pr(X'b_0 < 0) > 0$.
\end{enumerate}
\end{as}
%%%%%%%%%%%%%%%%%%%%%%%%%%%%%%%%%%%%%%%%%

Assumption \ref{as:bayes} is the standard assumption employed in the literature on cost-sensitive binary classification, and in particular, maximum score estimation.
Assumption \ref{as:bayes} \eqref{as:bayes:correct} requires the linear Bayes boundary,
that is, the linear classification $x \mapsto x'b$ is correctly specified with a representative \textit{true} parameter vector $b=b_0$.
In the literature on maximum score estimation, this assumption is essentially equivalent to what is known as the conditional median restriction \eqref{eq:conditional_median}.\footnote{\label{footnote:conditional_median}Indeed, in the threshold crossing model \eqref{eq:threshold_crossing}, the conditional median restriction \eqref{eq:conditional_median} implies Assumption \ref{as:bayes} \eqref{as:bayes:correct} under the regularity conditions that $\Pr(X'b_0=0)=0$ and $F_{\varepsilon|X}( \ \cdot \ |X)$ is strictly increasing at 0 almost surely.}
Assumption \ref{as:bayes} \eqref{as:bayes:nondegeneracy} requires nondegeneracy of the Bayes boundary, that is, there are non-trivial sub-populations on both sides across the classification boundary.

Let $Q_0$ denote the maximum score objective function defined in \eqref{eq:q0}.
Let $Q_\phi$ denote the surrogate objective function defined in \eqref{eq:q_phi}.
The following lemma shows that surrogate maximum score solutions can characterize the original maximum score solutions up to scale under two key conditions.

%%%%%%%%%%%%%%%%%%%%%%%%%%%%%%%%%%%%%%%%%
\begin{lm}[Validity of the Surroate Maximum Score]\label{lm:parallel}
Suppose that Assumption \ref{as:bayes} and the following two conditions are satisfied.
\begin{enumerate}[(L.\ref{lm:parallel}.1)]
\item[(L.\ref{lm:parallel}.1)] If $b_1,b_2 \in B$ are not parallel, then
$\Pr(\1\{X'b_1 \ge 0\} \neq \1\{X'b_2 \ge 0\}) > 0$.
\item[(L.\ref{lm:parallel}.2)] Every $b_\phi \in \arg\max_{b \in B} Q_\phi(b)$ satisfies $\1\{X'b_\phi \ge 0\} = \1\{\eta(X) \ge 1/2\}$ a.s.
\end{enumerate}
Then, for every $b_\phi \in \arg\max_{b \in B} Q_\phi(b)$, there exists a scalar $c>0$ such that $b_\phi = cb_0$, and $\arg\max_{b \in B} Q_\phi(b) \subseteq \arg\max_{b \in B} Q_0(b)$.
\end{lm}
%%%%%%%%%%%%%%%%%%%%%%%%%%%%%%%%%%%%%%%%%

Proof of this lemma is provided in Appendix \ref{sec:lm:parallel}.
Here, in the main text, we sketch the essence of the proof focusing on the roles played by the two key conditions, (L.\ref{lm:parallel}.1) and (L.\ref{lm:parallel}.2).

By way of contradiction, suppose that $b_\phi \in \arg\max_{b \in B} Q_\phi(b)$ does not parallel $b_0$.
Then, Condition (L.\ref{lm:parallel}.1) implies
$$
\Pr(\1\{X'b_\phi \ge 0\} \neq \1\{X'b_0 \ge 0\} ) > 0.
$$
On the other hand, Assumption \ref{as:bayes} \eqref{as:bayes:correct} and Condition (L.\ref{lm:parallel}.2) imply
$$
\1\{X'b_\phi \ge 0\} = \1\{X'b_0 \ge 0\} \qquad\text{a.s.},
$$
a contradiction.
Hence, $b_\phi = c b_0$ must hold for some scalar $c \neq 0$. 
The formal proof in the appendix further establishes $c>0$.

That is, any surrogate solution $b_\phi \in \arg\max_{b \in B} Q_\phi$ should coincide with the representative true parameter vector $b_0$ up to scale $c$, and therefore, it should belong to the set of solutions to the original maximum score problem.

%%%%%%%%%%%%%%%%%%%%%%%%%%%%%%%%%%%%%%%%%
\subsection{Existence and Uniqueness of the Surrogate Maximum Score Solution}\label{sec:existence_uniqueness}
%%%%%%%%%%%%%%%%%%%%%%%%%%%%%%%%%%%%%%%%%

Lemma \ref{lm:parallel} would be vacuous if the set $\arg\max_{b \in B} Q_\phi(b)$ of the surrogate maxima were empty. 
We are going to show that it is not empty under a suitable condition.
Furthermore, under additional conditions, it is not only nonempty, but also a singleton.
In other words, we etablish that the surrogate maximum score problem guarantees a unique solution, unlike the original maximum score problem.
To this goal, however, we need to choose a reasonable surrogate score function $\phi$ as formally outlined in the following assumption. 

%%%%%%%%%%%%%%%%%%%%%%%%%%%%%%%%%%%%%%%%%
\begin{as}[Surrogate Score Function]\label{as:surrogate}
$\phi$: $\mathbb{R} \to \mathbb{R}$ is strictly concave, strictly increasing and differentiable at 0 with $\phi'(0)>0$.
\end{as}
%%%%%%%%%%%%%%%%%%%%%%%%%%%%%%%%%%%%%%%%%
To satisfy Assumption \ref{as:surrogate}, we can take $\phi$ to be the negative of a common loss function $\ell$, that is, $\phi(u)=-\ell(u)$, where $\ell$ is strictly convex, strictly decreasing, and differentiable at $0$ with $\ell'(0)<0$. 
Examples include the following:
\begin{enumerate}[(1)]
	\item \textbf{Logistic loss}: $\ell(u)=\frac{1}{a} \log \left(1+e^{-a u}\right), a>0$;
	\item \textbf{Pseudo-Huber loss}: $\ell(u)=\sqrt{a^2+u^2}-u, a>0$;
	\item \textbf{Probit loss}: $\ell(u)=-\log\Phi(au)$, $a>0$, where $\Phi(\cdot)$ denotes the cumulative distribution function of the standard normal distribution.
\end{enumerate}

On the other hand, our requirements rule out other popular loss functions, such as the hinge loss, the ReLU loss, and square loss functions.
Even though the exponential loss function satisfies Assumption \ref{as:surrogate}, it does not generally satisfy an additional assumption to be required later (Assumption \ref{as:asymptotic} \eqref{as:asymptotic4}), and hence we desist from listing it as an example here.

In addition to the above requirements for surrogate score function $\phi$, we also impose the following regularity conditions.
%%%%%%%%%%%%%%%%%%%%%%%%%%%%%%%%%%%%%%%%%
\begin{as}[Regularity Conditions]\label{as:moments}
The following conditions hold.
\begin{enumerate}[(i)]
\item\label{as:moments:bounded} $\E\left[\sup_{b \in B} (|\phi(X'b)| + |\phi(-X'b)|)\right] < \infty$.
\item\label{as:moments:convex} $B = \{b \in \mathbb{R}^d : \|b\| \leq R\}$ for some $R \in (0,\infty)$.
\item\label{as:moments:index} For every $v \neq 0$ in $\mathbb{R}^d$, $\Pr(X'v \neq 0) > 0$.
\end{enumerate}
\end{as}
%%%%%%%%%%%%%%%%%%%%%%%%%%%%%%%%%%%%%%%%%
Assumption \ref{as:moments} \eqref{as:moments:bounded} requires the bounded moments of the surrogate score, and is used along with Assumption \ref{as:surrogate} to ensure the continuity of the surrogate objective function $Q_\phi$.
Assumption \ref{as:moments} \eqref{as:moments:convex} is under a researcher's control, and implies the convexity of the parameter set $B$ in particular.
For the sake of Lemma \ref{lm:existence_uniqueness} below, only the concavity of $B$ is needed. 
We state this particular form of $B = \{b \in \mathbb{R}^d : \|b\| \leq R\}$ in Assumption \ref{as:moments} \eqref{as:moments:convex} for the purpose of proving Proposition \ref{pp:surrogate} in Section \ref{sec:sufficient2} later.
Assumption \ref{as:moments} \eqref{as:moments:index} requires that the classifier $X'v$ is not degenerate for any possible parameter vector $v \neq 0$.
Assumption \ref{as:moments} \eqref{as:moments:convex}--\eqref{as:moments:index} serve to ensure the strict concavity of the surrogate objective function $Q_\phi$ on $B$.
With the continuity and strict concavity of $Q_\phi$ established as such, we can establish the following existence and uniqueness properties.

%%%%%%%%%%%%%%%%%%%%%%%%%%%%%%%%%%%%%%%%%
\begin{lm}[Existence and Uniqueness of the Surrogate Maximum]\label{lm:existence_uniqueness}
If Assumptions \ref{as:bayes}, \ref{as:surrogate}, and \ref{as:moments} \eqref{as:moments:bounded} hold, then a solution to $\max_{b \in B} Q_\phi(b)$ exists.
If Assumption \ref{as:moments} \eqref{as:moments:convex}--\eqref{as:moments:index} hold in addition, then $\max_{b \in B} Q_\phi(b)$ admits a unique solution.
\end{lm}
%%%%%%%%%%%%%%%%%%%%%%%%%%%%%%%%%%%%%%%%%

Proof of Lemma \ref{lm:existence_uniqueness} is provided in Appendix \ref{sec:lm:existence_uniqueness}.
The continuity of $\phi$ implied by Assumption \ref{as:surrogate} does not necessarily guarantee the continuity of $Q_\phi$. 
To ensure the continuity of $Q_\phi$, we use the bounded moment condition of Assumption \ref{as:moments} \eqref{as:moments:bounded} in addition so that score can be shown to be integrable to invoke the dominated convergence theorem.
Then, the standard argument via Weierstrass Theorem yields the existence result.
The uniqueness follows by the strict concavity of $Q_\phi$ on $B$, which holds under Assumption \ref{as:moments} \eqref{as:moments:convex}--\eqref{as:moments:index}.

%%%%%%%%%%%%%%%%%%%%%%%%%%%%%%%%%%%%%%%%%
\subsection{Summary of the Main Theoretical Results}
%%%%%%%%%%%%%%%%%%%%%%%%%%%%%%%%%%%%%%%%%

Putting \blue{Lemma} \ref{lm:parallel} and \ref{lm:existence_uniqueness} together, we obtain the following theorem.

%%%%%%%%%%%%%%%%%%%%%%%%%%%%%%%%%%%%%%%%%
\begin{thm}\label{thm:main}
If Assumptions \ref{as:bayes}, \ref{as:surrogate}, and \ref{as:moments} hold, then there exists a unique solution $b_\phi = \arg\max_{b \in B}Q_\phi(b)$.
Suppose that the following two conditions are satisfied in addition.
\begin{enumerate}[(T.\ref{thm:main}.1)]
\item[(T.\ref{thm:main}.1)] If $b_1,b_2 \in B$ are not parallel, then
$\Pr(\1\{X'b_1 \ge 0\} \neq \1\{X'b_2 \ge 0\}) > 0$.
\item[(T.\ref{thm:main}.2)] $\1\{X'b_\phi \ge 0\} = \1\{\eta(X) \ge 1/2\}$ a.s.
\end{enumerate}
Then, there exists a scalar $c>0$ such that $b_\phi = cb_0$, and $b_\phi \in \arg\max_{b \in B} Q_0(b)$.
\end{thm}
%%%%%%%%%%%%%%%%%%%%%%%%%%%%%%%%%%%%%%%%%

This theorem states that the surrogate maximum score can uniquely characterize a solution to the original maximum score under the two key conditions, (T.\ref{thm:main}.1) and (T.\ref{thm:main}.2).
Provided that these two conditions are satisfied, the true representative parameter $b_0$ can  be identified by the unique surrogate solution $b_\phi$ to $\max_{b \in B}Q_\phi(b)$ up to scale $c>0$.

%%%%%%%%%%%%%%%%%%%%%%%%%%%%%%%%%%%%%%%%%
\section{Primitive Sufficient Conditions and Examples}\label{sec:primitive}
%%%%%%%%%%%%%%%%%%%%%%%%%%%%%%%%%%%%%%%%%

The two key conditions, (T.\ref{thm:main}.1) and (T.\ref{thm:main}.2), stated in Theorem \ref{thm:main} are nontrivial.
Yet, they are satisfied by a wide family of distributions of $X$.
This section investigates primitive sufficient conditions for these high-level conditions, and list some families of distributions $X$ that satisfy these conditions.

For ease of exposition, we focus on the case with no shifting, where the model involves no intercept and the center of the distribution of $X$ is zero.
However, we emphasize that these locational restrictions can be relaxed at the expense of more sophisticated writings.

%%%%%%%%%%%%%%%%%%%%%%%%%%%%%%%%%%%%%%%%%
\subsection{Sufficient Condition for Condition (T.\ref{thm:main}.1)}\label{sec:sufficient1}
%%%%%%%%%%%%%%%%%%%%%%%%%%%%%%%%%%%%%%%%%

We first establish a sufficient condition for Condition (T.\ref{thm:main}.1).
To this end, we introduce some notation.
Let $\lambda_d$ denote the Lebesgue measure on $\mathbb{R}^d$, and let $B_r(x) = \{\xi \in \mathbb{R}^d : \|\xi-x\| < r\}$ denote the open ball of radius $r$ around $x \in \mathbb{R}^d$.
%%%%%%%%%%%%%%%%%%%%%%%%%%%%%%%%%%%%%%%%%
\begin{as}[Local Full Support]\label{as:local_full_support}
There exists $r>0$ such that 
$\lambda_d(A)>0$ implies $\Pr(X \in A) > 0$
for every Borel set $A \subseteq B_r(0)$.
\end{as}
%%%%%%%%%%%%%%%%%%%%%%%%%%%%%%%%%%%%%%%%%
Assumption \ref{as:local_full_support} requires that the distribution of $X$ places positive probability
on every subset of some open neighborhood of the origin that has positive
Lebesgue measure.
This assumption can be satisfied, for example, if $X$ has a probability density function $f_X$ satisfying $f_X(x)>0$ for almost every $x \in B_r(0)$.

The following proposition shows that Assumption \ref{as:local_full_support} is sufficient for Condition (T.\ref{thm:main}.1).
%%%%%%%%%%%%%%%%%%%%%%%%%%%%%%%%%%%%%%%%%
\begin{pp}[Sufficient Condition for Condition (T.\ref{thm:main}.1)]\label{pp:nonzero_probability}
Suppose that Assumption \ref{as:local_full_support} holds.
Then, Condition (T.\ref{thm:main}.1) in the statement of Theorem \ref{thm:main} is satisfied.
\end{pp}
%%%%%%%%%%%%%%%%%%%%%%%%%%%%%%%%%%%%%%%%%

Proof is provided in Appendix \ref{sec:pp:nonzero_probability}.
The core of its proof shows that there exists a nonempty open set $A$ that is contained in both the set
$$
D = \{x \in \mathbb{R}^d : \1\{x'b_1 \ge 0\} \neq \1\{x'b_2 \ge 0\}\}
$$
and $B_r(0)$.
Hence, $\Pr(D)>0$ results by Assumption \ref{as:local_full_support}, and Condition (T.\ref{thm:main}.1) is satisfied consequently.

%%%%%%%%%%%%%%%%%%%%%%%%%%%%%%%%%%%%%%%%%
\subsection{Sufficient Conditions for Condition (T.\ref{thm:main}.2)}\label{sec:sufficient2}
%%%%%%%%%%%%%%%%%%%%%%%%%%%%%%%%%%%%%%%%%

We next establish sufficient conditions for Condition (T.\ref{thm:main}.2).
These conditions are the single-index assumptions stated below.
%%%%%%%%%%%%%%%%%%%%%%%%%%%%%%%%%%%%%%%%%
\begin{as}[Single Index]\label{as:index}
Denote $T = X'b_0$. The following conditions hold.
\begin{enumerate}[(i)]
\item\label{as:index:bounded_moment} $\E[|T|] < \infty$.
\item\label{as:index:single_index} There exists a measurable function $h: \mathbb{R} \to [0,1]$ such that $\eta(X) = h(T)$, $h(t)>1/2$ for all $t>0$, and $h(t)<1/2$ for all $t<0$.
%$h$ is strictly increasing, and $h(0) = 1/2$.
\item\label{as:index:linear} For every $b \in \mathbb{R}^d$, there exists $a_b \in \mathbb{R}$ such that $a_b b_0 \in B$ and $\E[X'b|T] = a_b T$ a.s. 
\end{enumerate}
\end{as}
%%%%%%%%%%%%%%%%%%%%%%%%%%%%%%%%%%%%%%%%%
Assumption \ref{as:index} \eqref{as:index:bounded_moment} requires that the single index $T$ has a bounded moment, which is a quite mild assumption.
A sufficient condition is $\E[\|X\|]<\infty$.
Assumption \ref{as:index} \eqref{as:index:single_index} requires that the conditional choice probability is strictly increasing as a function of the single index $X'b_0$.
This index assumption is also assumed in the literature \citep[e.g.,][]{klein1993efficient}.
In particular, it is satisfied under the threshold-crossing model \eqref{eq:threshold_crossing} with the conditional median restriction \eqref{eq:conditional_median} provided that $F_{\varepsilon|X}$ is strictly increasing at 0.
Assumption \ref{as:index} \eqref{as:index:linear} requires that $X'b$ linearly project on $T$.
Unlike the previous two parts of the assumption, this part is a nontrivial assumption.
With this said, this nontrivial part can be satisfied by a wide class of distributions of $X$, as formally explored in Section \ref{sec:examples_distribution} with examples.

The following theorem shows that Assumption \ref{as:index} serves as a primitive sufficient condition for the high-level condition (T.\ref{thm:main}.2) in the statement of Theorem \ref{lm:parallel}.

%%%%%%%%%%%%%%%%%%%%%%%%%%%%%%%%%%%%%%%%%
\begin{pp}[Sufficient Condition for Condition (T.\ref{thm:main}.2)]\label{pp:surrogate}
Suppose that Assumptions \ref{as:bayes}, \ref{as:surrogate}, \ref{as:moments}, and \ref{as:index} hold.
Then, Condition (T.\ref{thm:main}.2) in the statement of Theorem \ref{thm:main} is satisfied.
\end{pp}
%%%%%%%%%%%%%%%%%%%%%%%%%%%%%%%%%%%%%%%%%

Proof is provided in Appendix \ref{sec:pp:surrogate}.
Key to the proof is Assumption \ref{as:index} \eqref{as:index:linear}, which allows for conditional Jensen's inequality to derive
$
Q_\phi(b) \le Q_\phi(a_b b_0)
$
for every $b \in B$.
As already mentioned, this key assumption is nontrivial.
The following subsection discusses examples of distributions satisfying Assumption \ref{as:index} \eqref{as:index:linear}, as well as Assumptions \ref{as:local_full_support} and \ref{as:index} \eqref{as:index:bounded_moment}.

%%%%%%%%%%%%%%%%%%%%%%%%%%%%%%%%%%%%%%%%%
\subsection{Examples of Distribution Families Satisfying the Conditions}\label{sec:examples_distribution}
%%%%%%%%%%%%%%%%%%%%%%%%%%%%%%%%%%%%%%%%%

This section discusses specific distribution families of $X$ that satisfy Assumptions \ref{as:local_full_support}, \ref{as:index} \eqref{as:index:bounded_moment}, and \ref{as:index} \eqref{as:index:linear}.
We introduce the following definition of elliptically symmetric distributions.
%%%%%%%%%%%%%%%%%%%%%%%%%%%%%%%%%%%%%%%%%
\begin{df}
A random vector $X$ is elliptically distributed if its probability density function $f_X$ takes the form
$$
f(x) = |\Sigma|^{-1/2} g\left((x-\mu)' \Sigma^{-1} (x-\mu)\right).
$$
We write $X \sim ES(\mu,\Sigma,g)$ if $X$ has this form of $f_X$.
\end{df}
%%%%%%%%%%%%%%%%%%%%%%%%%%%%%%%%%%%%%%%%%

If $X \sim ES(0,\Sigma,g)$, then we can derive
$$
\E\left[ X | b_0'X = t \right]
=
\frac{\Sigma b_0}{b_0' \Sigma b_0} t,
$$
hence satisfying Assumption \ref{as:index} \eqref{as:index:linear}.

Examples of elliptically symmetric distributions include the multivariate normal, multivariate $t$, and multivariate Laplace distributions, among others. More generally, Gaussian scale mixtures of the form
$
X = \sqrt{S}\,Z,
$
where $Z \sim N(0,\Sigma)$ and $S>0$ is independent of $Z$, also possess the elliptical symmetry property. This class is very wide, but includes, for example, the variance--gamma, generalized hyperbolic, normal--inverse--Gaussian, and slash distributions.
Furthermore, all these families of distributions satisfy Assumption \ref{as:index} \eqref{as:index:bounded_moment} under Assumption \ref{as:bayes}.
All these families of distribution also satisfy Assumption \ref{as:local_full_support} as well.
Hence, these distributions serve as examples that satisfy the high-level conditions, (T.\ref{thm:main}.1) and (T.\ref{thm:main}.2), stated in Theorem \ref{thm:main} under the threshold-crossing model \eqref{eq:threshold_crossing} with the conditional median restriction \eqref{eq:conditional_median}.

Finally, we emphasize that these distribution families serve only as illustrative examples satisfying the sufficient conditions. 
The sufficient conditions themselves are far from necessary, suggesting that a substantially broader class of distributions, including nonparametric families, may also satisfy them.

%%%%%%%%%%%%%%%%%%%%%%%%%%%%%%%%%%%%%%%%%
\subsection{Discrete $X$}
%%%%%%%%%%%%%%%%%%%%%%%%%%%%%%%%%%%%%%%%%

Section \ref{sec:primitive} has thus far focused on primitive sufficient conditions that entail continuous distributions, in order to facilitate clear discussions of commonly used distribution families such as the multivariate normal, multivariate $t$, and multivariate Laplace distributions, among others, as introduced in Section \ref{sec:examples_distribution} as specific examples satisfying the high-level conditions. 
However, these sufficient conditions are far from necessary. 

In particular, the high-level conditions (T.\ref{thm:main}.1)--(T.\ref{thm:main}.2) in Theorem \ref{thm:main} do not rule out discretely supported $X$, although, unlike the continuous case, there are no canonical stylized examples for multivariate discrete variables.
Specifically, the high-level Condition (T.\ref{thm:main}.1) can hold for discrete $X$ if the support of $X$ is sufficiently rich such that, for every pair of nonparallel $b_1, b_2 \in B$, there exists at least one support point $x$ with positive probability mass satisfying the inequality $\1\{x'b_1 \ge 0\} \neq \1\{x'b_2 \ge 0\}$. 
Furthermore, our primitive sufficient condition of the single-index structure for the high-level Condition (T.\ref{thm:main}.2) is not inherently tied to continuity of $X$.

%%%%%%%%%%%%%%%%%%%%%%%%%%%%%%%%%%%%%%%%%
\section{Surrogate Maximum Score Estimation \& Standard Inference}\label{sec:estimation}
%%%%%%%%%%%%%%%%%%%%%%%%%%%%%%%%%%%%%%%%%

%%%%%%%%%%%%%%%%%%%%%%%%%%%%%%%%%%%%%%%%%
\subsection{Estimator}
%%%%%%%%%%%%%%%%%%%%%%%%%%%%%%%%%%%%%%%%%

The conventional maximum score estimator is based on the sample counterpart of 
$\max_{b \in B} Q_0(b)$, where $Q_0$ is defined in \eqref{eq:q0}. 
Because of the presence of the indicator function, this method presents both practical and theoretical challenges. 
In particular, the resulting optimization problem is nonconvex, and the estimator converges at a slower-than-root-$n$ rate to a nonstandard limiting distribution.

Under the conditions discussed in Sections \ref{sec:main} and \ref{sec:primitive}, we argued that the surrogate objective $Q_\phi$ defined in \eqref{eq:q_phi} can be used in place of $Q_0$. 
Motivated by this observation, we propose the surrogate maximum score estimator
\begin{align}
&\widehat b = \arg\max_{b\in B}Q_{\phi, n}(b),
\qquad\text{where}\notag\\
&Q_{\phi, n}(b) = \mathbb{E}_n \big[ Y \cdot \phi(X'b) + (1-Y) \cdot \phi(-X'b) \big],
\label{eq:surrogate_sample_objective}
\end{align}
with $\mathbb{E}_n$ denoting the sample mean operator. 
The equality ``$=$'' in this definition characterizes $\widehat b$ uniquely, as the solution is shown to be unique in Section \ref{sec:existence_uniqueness}.

The surrogate score function $\phi$ is well-behaved, cf.~Assumption \ref{as:surrogate}.
Therefore, unlike the conventional maximum score estimator, the surrogate method in \eqref{eq:surrogate_sample_objective} admits a convex optimization problem, a unique solution, and a root-$n$ convergence rate to a standard limiting distribution. 
Although the underlying theoretical arguments are largely standard, the following subsection presents the asymptotic properties of the estimator \eqref{eq:surrogate_sample_objective} for the sake of completeness.

%%%%%%%%%%%%%%%%%%%%%%%%%%%%%%%%%%%%%%%%%
\subsection{Asymptotic Properties}\label{sec:asymptotic}
%%%%%%%%%%%%%%%%%%%%%%%%%%%%%%%%%%%%%%%%%

The uniqueness of the solution $b_\phi$, as established in Theorem \ref{thm:main}, together with the local curvature of the population objective, as described below, provides the key ingredients for deriving its statistical properties.
For convenience of writing, we write
$\ell_\phi(Z_i, b):=Y_i \phi(X_i^{\prime} b)+(1-Y_i) \phi(-X_i^{\prime} b)$
and
$\psi(Z_i, b):=\nabla_b \ell_\phi(Z, b)$.

%%%%%%%%%%%%%%%%%%%%%%%%%%%%%%%%%%%%%%%%%
\begin{as}\label{as:asymptotic}${}$
\begin{enumerate}[(i)]
\item\label{as:asymptotic1} $Z_i:=(Y_i, X^\prime_i)^\prime$ are i.i.d. and the support of $X_i$ is a set $\mathcal{X} \subset \mathbb{R}^d$.

\item \label{as:asymptotic0} The parameter $b_\phi$ belongs to a compact set $B \subset \mathbb{R}^d$;

%\item \label{as:asymptotic2} $\mathbb{E}\Vert X_i\Vert ^{4}<\infty$.

\item \label{as:asymptotic3} $\ell_\phi$ is twice continuously differentiable;
$H:=\nabla^2_b Q_\phi(b_\phi)$ is negative definite; and
$\Omega:=\mathbb{E}[\psi(Z_i, b_\phi) \psi(Z_i, b_\phi)^{\prime}]$ is finite and positive definite.

\item \label{as:asymptotic4}
$\mathbb{E}[F_0(Z)^2]<\infty$, $\mathbb{E}[F_1(Z)^2]<\infty$, and $\mathbb{E}[F_2(Z)^2]<\infty$, where
$F_0(Z):=\sup _{b \in B}\vert\ell_\phi(Z, b)\vert$,
$F_1(Z):=\sup _{b \in B}\Vert\nabla_b \ell_\phi(Z, b)\Vert$, and 
$F_2(Z):=\sup _{b \in B}\Vert\nabla_b^2 \ell_\phi(Z, b)\Vert$.
\end{enumerate}
\end{as}
%%%%%%%%%%%%%%%%%%%%%%%%%%%%%%%%%%%%%%%%%

Assumption \ref{as:asymptotic} \eqref{as:asymptotic1} imposes an i.i.d.\ sampling framework. 
This assumption is not essential and can be extended to accommodate time-series or spatial dependence.
Assumption \ref{as:asymptotic} \eqref{as:asymptotic0} requires the compactness of the parameter set.
Assumption \ref{as:asymptotic} \eqref{as:asymptotic3} is compatible with the smooth surrogate function $\phi$ and the uniqueness of the maximizer as guaranteed by Lemma \ref{lm:existence_uniqueness}. In particular, it can be satisfied by the specific loss functions introduced as examples in Section \ref{sec:existence_uniqueness}.

The envelope conditions in Assumption \ref{as:asymptotic} \eqref{as:asymptotic4} are stated at a high level to maintain generality. Lower-level sufficient conditions are provided in Section \ref{sec:as:asymptotic4}. In particular, we show that the conventional bounded fourth-moment condition $\E\|X\|^4 < \infty$ is sufficient for Assumption \ref{as:asymptotic} \eqref{as:asymptotic4} in the context of the specific loss functions presented as examples in Section \ref{sec:existence_uniqueness}.

The following corollary formalizes the root-$n$ consistency and asymptotic normality of the surrogate maximum score estimator $\widehat{b}$.

%%%%%%%%%%%%%%%%%%%%%%%%%%%%%%%%%%%%%%%%%
\begin{cl}[Root-$n$ Asymptotic Normality]\label{cl:root-n}
If Assumption \ref{as:asymptotic} holds in addition to the conditions stated in Theorem \ref{thm:main}, then
\begin{align*}
\widehat{b} \rightarrow_{p} b_\phi
~~\text{and}~~
\sqrt{n}(\widehat{b}-b_\phi) \rightarrow_{d} \mathcal{N}(0,~H^{-1} \Omega H^{-1}) .
\end{align*}
\end{cl}
%%%%%%%%%%%%%%%%%%%%%%%%%%%%%%%%%%%%%%%%%

Corollary \ref{cl:root-n} establishes that the surrogate maximum score estimator $\widehat b$ attains the standard root-$n$ rate of convergence and is asymptotically normal. 
This stands in contrast to the classical maximum score estimator, which exhibits a cube-root convergence rate and a non-Gaussian limiting distribution due to the discontinuity of its objective function. 

The key feature underlying this improvement is the smoothness of the surrogate loss $\phi(\cdot)$, which ensures that the population objective $Q_\phi(b)$ is twice continuously differentiable and admits a local quadratic expansion around its maximizer $b_\phi$. 
As a consequence, the asymptotic normality of $\widehat b$ permits the use of conventional inference procedures, including bootstrap methods, thereby facilitating standard Gaussian-based inference.

Section \ref{sec:simulation} demonstrates these theoretical properties—namely, root-$n$ consistency and asymptotic normality—through simulation studies.

%%%%%%%%%%%%%%%%%%%%%%%%%%%%%%%%%%%%%%%%%
\subsection{Primitive Sufficient Conditions for Assumption \ref{as:asymptotic} \eqref{as:asymptotic4}}\label{sec:as:asymptotic4}
%%%%%%%%%%%%%%%%%%%%%%%%%%%%%%%%%%%%%%%%%

The envelope conditions in Assumption \ref{as:asymptotic} \eqref{as:asymptotic4} are stated at a high level to preserve generality. In this section, we provide lower-level sufficient conditions tailored to the three loss functions introduced in Section \ref{sec:existence_uniqueness}.
Namely, it is enough to assume the conventional bounded fourth-moment condition $\E\|X\|^4 < \infty$.

Let the parameter space $B \subset \mathbb{R}^d$ be compact, and suppose that $\phi$ and its derivatives satisfy polynomial growth:
\begin{align*}
|\phi(u)| \lesssim 1 + |u|^p, 
\quad
|\phi'(u)| \lesssim 1 + |u|^q,
\quad \text{and} \quad
|\phi''(u)| \lesssim 1 + |u|^r
\end{align*}
for some constants $p, q, r \ge 0$.
Then, the envelope functions satisfy the bounds
\begin{align*}
F_0(Z) \lesssim 1 + \|X\|^p, 
\quad
F_1(Z) \lesssim \|X\| \bigl(1 + \|X\|^q\bigr),
\quad \text{and} \quad
F_2(Z) \lesssim \|X\|^2 \bigl(1 + \|X\|^r\bigr).
\end{align*}
In light of these bounds, we provide primitive sufficient conditions for Assumption \ref{as:asymptotic} \eqref{as:asymptotic4} for each of the loss functions considered below.

\begin{enumerate}[(1)]
\item \textbf{Logistic loss:} 
$\phi(u) = -\frac{1}{a}\log(1+e^{-au}), a>0$.
The polynomial growth conditions are satisfied with the exponents $p=q=r=0$, yielding
\begin{align*}
F_0(Z)\lesssim 1,
\quad 
F_1(Z)\lesssim \Vert X\Vert,
\quad\text{and}\quad
F_2(Z)\lesssim \Vert X\Vert^2.
\end{align*}
Hence, Assumption \ref{as:asymptotic} \eqref{as:asymptotic4} is satisfied if $\mathbb{E}\Vert X\Vert^4<\infty$.

\item \textbf{Pseudo-Huber loss:} 
$\phi(u) = -\sqrt{a^2 + u^2} + u$, $a>0$. 
The polynomial growth conditions are satisfied with the exponents $p=1$ and $q=r=0$, yielding
\begin{align*}
F_0(Z)\lesssim 1+\Vert X\Vert,
\quad
F_1(Z)\lesssim \Vert X\Vert,
\quad\text{and}\quad
F_2(Z)\lesssim \Vert X\Vert^2. 
\end{align*}
Hence, Assumption \ref{as:asymptotic} \eqref{as:asymptotic4} is satisfied if $\mathbb{E}\Vert X\Vert^4<\infty$.

\item \textbf{Probit loss:} 
$\phi(u)=-\log \Phi(au)$. 
The polynomial growth conditions are satisfied with the exponents $p=2$, $q=1$, and $r=0$, yielding
\begin{align*}
	F_0(Z)\lesssim 1+\Vert X\Vert^2,
	\quad
	F_1(Z)\lesssim \Vert X\Vert(1+\Vert X\Vert),
	\quad\text{and}\quad 
	F_2(Z)\lesssim \Vert X\Vert^2.
\end{align*}
Hence, Assumption \ref{as:asymptotic} \eqref{as:asymptotic4} is satisfied if $\mathbb{E}\Vert X\Vert^4<\infty$.

\end{enumerate}

All the three loss functions introduced in Section \ref{sec:existence_uniqueness} as examples satisfying the high-level conditions in Assumption \ref{as:surrogate} also satisfy the high-level condition in Assumption \ref{as:asymptotic} \eqref{as:asymptotic4}. %under the conventional bounded fourth-moment condition $\E\|X\|^4 < \infty$.

On the other hand, not all loss functions that satisfy Assumption \ref{as:surrogate} meet the requirement in Assumption \ref{as:asymptotic} \eqref{as:asymptotic4}. For instance, the exponential loss satisfies Assumption \ref{as:surrogate}, but it cannot be bounded by any polynomial and therefore fails to satisfy Assumption \ref{as:asymptotic} \eqref{as:asymptotic4}. This is the primary reason why it was not included as an example in Section \ref{sec:existence_uniqueness}. 

%%%%%%%%%%%%%%%%%%%%%%%%%%%%%%%%%%%%%%%%%
\subsection{Variance Estimation}\label{sec:variance}
%%%%%%%%%%%%%%%%%%%%%%%%%%%%%%%%%%%%%%%%%

In light of the root-$n$ asymptotic normality established in Corollary \ref{cl:root-n}, we can conduct standard inference by computing the analytic variance estimate to obtain standard errors, constructing confidence intervals, comparing $t$-statistics with normal critical values, and computing $p$-values. 

This feature of the surrogate maximum score method offers a substantial advantage: it can be readily implemented by empirical researchers familiar with standard asymptotics and is compatible with widely used statistical software packages, such as Stata, whose default output assumes a normal limiting distribution.

To this end, this subsection presents the estimation of the asymptotic variance $H^{-1} \Omega H^{-1}$ and establishes its consistency. 
This can be implemented via a straightforward plug-in approach. 
Let $\widehat{b}$ be the surrogate estimator of $b_\phi$, and define the sample Hessian and variance estimator as
\begin{align*}
\widehat{H} := \nabla_b^2 Q_n(\widehat{b}) 
\quad \text{and} \quad 
\widehat{\Omega} := \frac{1}{n} \sum_{i=1}^n \psi(Z_i, \widehat{b}) \, \psi(Z_i, \widehat{b})^{\prime},
\end{align*}
where $Q_n(b) = \frac{1}{n} \sum_{i=1}^n \ell_\phi(Z_i, b)$ and $\psi(Z_i, b) = \nabla_b \ell_\phi(Z_i, b)$. 
Then, a natural estimator of the asymptotic variance is given by
\begin{align}\label{eq:var2}
\widehat{V} := \widehat{H}^{-1} \widehat{\Omega} \widehat{H}^{-1}.
\end{align}
The following corollary establishes the consistency of this variance estimator.

%%%%%%%%%%%%%%%%%%%%%%%%%%%%%%%%%%%%%%%%%
\begin{cl}\label{cl:var}
Suppose that the conditions stated in Corollary \ref{cl:root-n} hold. Then,
\begin{align*}
\widehat{V} \rightarrow_{p} H^{-1} \Omega H^{-1}.
\end{align*}
\end{cl}
%%%%%%%%%%%%%%%%%%%%%%%%%%%%%%%%%%%%%%%%%

The matrix $\widehat{V}$ provides a consistent estimator of the asymptotic covariance of $\sqrt{n}(\widehat{b}-b_\phi)$, and can thus be used to construct standard errors and confidence intervals. 

In particular, the asymptotic normality result in Corollary \ref{cl:root-n}, together with the consistency of the variance estimator in Corollary \ref{cl:var}, implies that for any fixed vector $a \in \mathbb{R}^d$,
\begin{align*}
T_n:=\frac{\sqrt{n}\, a^{\prime}(\widehat{b}-b_\phi)}{\sqrt{a^{\prime} \widehat{V} a}} \rightarrow_{d} \mathcal{N}(0,1),
\end{align*}
which justifies the use of $\widehat{V}$.

While these theoretical results and properties are standard, such “standard” results are rather unexpected in the context of maximum score estimation. We therefore dare to emphasize these points here.

%%%%%%%%%%%%%%%%%%%%%%%%%%%%%%%%%%%%%%%%%
\subsection{Nonparametric Bootstrap with Asymptotic Refinement}\label{sec:bootstrap}
%%%%%%%%%%%%%%%%%%%%%%%%%%%%%%%%%%%%%%%%%

A large literature on the maximum score problem has focused on resampling-based inference, motivated by the well-known failure of the standard nonparametric bootstrap; see Section \ref{sec:introduction} for details.

Once the maximum score problem is rendered ``standard'' via the surrogate score, however, the conventional nonparametric bootstrap becomes applicable and can deliver asymptotic refinements \citep{horowitz2001bootstrap,horowitz2019bootstrap}, often leading to substantial finite-sample improvements. This section formalizes this bootstrap result.

Let $a \in \mathbb{R}^d$ be a fixed vector, and define the scalar parameter $\theta := a^\prime b_\phi$ of interest with estimator $\widehat{\theta} := a^\prime \widehat{b}$. Let $\widehat{V}$ be the consistent variance estimator \eqref{eq:var2}, and define the standard error
\begin{align*}
\widehat{\sigma}^2 := a^\prime \widehat{V} a.
\end{align*}
The corresponding normalized test statistic is given by
\begin{align*}
T_n = \frac{\sqrt{n}\bigl(a^\prime \widehat{b} - a^\prime b_\phi\bigr)}{\widehat{\sigma}}.
\end{align*}

To implement the bootstrap, draw $S$ independent bootstrap samples $\{Z_i^{\ast(s)}\}_{i=1}^n$ with replacement from the empirical distribution. For each $s \in [S]$, compute the bootstrap estimator
\begin{align*}
\widehat{b}^{\ast(s)} := \arg\max_{b \in B} \frac{1}{n} \sum_{i=1}^n \ell_\phi(Z_i^{\ast(s)}, b),
\end{align*}
and define $\widehat{\theta}^{\ast(s)} := a^\prime \widehat{b}^{\ast(s)}$. Let $\widehat{V}^{\ast(s)}$ denote the bootstrap analogue of $\widehat{V}$, and define
\begin{align*}
\widehat{\sigma}^{\ast(s)2} := a^\prime \widehat{V}^{\ast(s)} a.
\end{align*}
The bootstrap studentized statistic is given by
\begin{align*}
T_n^{\ast(s)} := \frac{\sqrt{n}\bigl(a^\prime \widehat{b}^{\ast(s)} - a^\prime \widehat{b}\bigr)}{\widehat{\sigma}^{\ast(s)}}.
\end{align*}

Let $\widehat{q}_{\alpha}^\ast$ denote the empirical $\alpha$-quantile of $\{T_n^{\ast(s)}\}_{s=1}^S$. A $(1-\alpha)$ confidence interval for $\theta = a^\prime b_\phi$ is then given by
\begin{align*}
\left[
a^\prime \widehat{b} - \frac{\widehat{q}_{1-\alpha/2}^\ast \, \widehat{\sigma}}{\sqrt{n}},
\quad
a^\prime \widehat{b} - \frac{\widehat{q}_{\alpha/2}^\ast \, \widehat{\sigma}}{\sqrt{n}}
\right].
\end{align*}

The following corollary establishes the asymptotic validity and higher-order accuracy of the proposed nonparametric bootstrap procedure.

\begin{as}\label{as:bootstrap}
	\begin{enumerate}[(i)]
		\item\label{as:bootstrap1} The Cram\'er's condition,
		$
			\limsup_{\Vert t\Vert\to\infty}\vert \mathbb{E} e^{it^\prime W_i}\vert<1,
		$
		is satisfied for the characteristic function of the three-dimensional random vector
		\begin{align*}
		W_i:=(a^\prime H^{-1}\psi(Z_i,b_\phi),a^\prime H^{-1}(\psi(Z_i,b_\phi)\psi(Z_i,b_\phi)^\prime -\Omega)H^{-1}a,a^\prime H^{-1}(\nabla_b^2 \ell_\phi(Z_i, b)-H)Va)^\prime.
		\end{align*}
		\item\label{as:bootstrap2} Assume that the surrogate loss $\phi$ is three-times differentiable.
		\item\label{as:bootstrap3} Assume that $\mathbb{E}\Vert \psi(Z_i, b_{\phi})\Vert^6 < \infty$.
	\end{enumerate}
\end{as}

\begin{cl}\label{cl:studentized_bootstrap}
		Suppose that the conditions of Corollary \ref{cl:root-n} and Assumption \ref{as:bootstrap} hold. Then,
		\begin{align*}
			\sup_{t \in \mathbb{R}} 
			\left\vert
			\mathbb{P}^\ast(T_n^\ast \leq t) - \mathbb{P}(T_n \leq t)
			\right\vert
			= o(n^{-1/2}),
		\end{align*}
		where $\mathbb{P}^*(\cdot)$ denotes the conditional bootstrap probability over data. In particular, the studentized bootstrap achieves an asymptotic refinement over the normal approximation of Subsection \ref{sec:variance}, whose approximation error is of order $O(n^{-1/2})$.
\end{cl}

%%%%%%%%%%%%%%%%%%%%%%%%%%%%%%%%%%%%%%%%%
\section{Simulation Evidence}\label{sec:simulation}
%%%%%%%%%%%%%%%%%%%%%%%%%%%%%%%%%%%%%%%%%

While we have established the theory of root-$n$ asymptotic normality, the conventional theory for the maximum score estimator does not exhibit these desirable properties. 
Therefore, it is useful to provide numerical evidence in support of our theoretical predictions. 

This section presents simulation studies that demonstrate the root-$n$ convergence rate (Section \ref{sec:sim:root_n}), the normal limiting distribution (Section \ref{sec:sim:normal}), and the validity of inference methods based on these theoretical properties (Section \ref{sec:sim:inference}). 
We begin by presenting our simulation design in Section \ref{sec:design}.

%%%%%%%%%%%%%%%%%%%%%%%%%%%%%%%%%%%%%%%%%
\subsection{Simulation Design}\label{sec:design}
%%%%%%%%%%%%%%%%%%%%%%%%%%%%%%%%%%%%%%%%%

We generate $n$ independent copies of $(Y_i, X_i')'$ from the threshold-crossing model \eqref{eq:threshold_crossing}. 
Let 
\[
\Sigma=\begin{pmatrix}
	1 & 0.5\\
	0.5 & 1
\end{pmatrix},
\]
and consider the following three designs for $X_i$ in line with the discussion in Section \ref{sec:examples_distribution}:
\begin{enumerate}[(i)]
	\item \textbf{Normal distribution}: $X_i\sim N(0,\Sigma)$;
	\item \textbf{$t$ distribution ($t_5$)}: $X_i=Z_i/\sqrt{U_i/5}$, where $Z_i\sim N(0,\Sigma)$ and $U_i\sim\chi^2_5$ independently;
	\item \textbf{Laplace distribution}: $X_i=\sqrt{S_i}\,Z_i$, where $Z_i\sim N(0,\Sigma)$ and $S_i\sim \text{Exp}(1)$ independently.
\end{enumerate}
Let $\varepsilon_i \sim \text{Logistic}(0,1)$ be independent of $X_i$, and set $b_0 = (1,1)^\prime / \sqrt{2}$.

Our parameter of interest is a scalar reparameterization $\theta$ of the two-dimensional coefficient vector $b$ under the unit-norm normalization. Specifically, we write
\begin{align*}
	b(\theta)=
	\begin{pmatrix}
		\cos\theta\\ \sin\theta
	\end{pmatrix}, 
	\qquad 
	\|b(\theta)\|=1.
\end{align*}
Thus, estimating $\theta$ is equivalent to estimating the direction of $b$. 
Given the true parameter values $b_0 = \frac{1}{\sqrt{2}}(1,1)^\prime$, the corresponding true angle is $\theta_0 = \frac{\pi}{4}$. 
This parameterization is convenient because the threshold-crossing model \eqref{eq:threshold_crossing} yields observationally equivalent distributions under scaling of $b$, whereas $\theta$ provides a unique scalar target for reporting RMSE and plotting densities.
Note that this reparameterization is smooth, so the theoretical results established in Section \ref{sec:estimation} extend to $\theta$.

%%%%%%%%%%%%%%%%%%%%%%%%%%%%%%%%%%%%%%%%%
\subsection{Evidence of the Root-$n$ Consistency}\label{sec:sim:root_n}
%%%%%%%%%%%%%%%%%%%%%%%%%%%%%%%%%%%%%%%%%

This section presents evidence of the root-$n$ convergence rate in support of Corollary \ref{cl:root-n}.
For each design, we conduct $10{,}000$ Monte Carlo replications for each sample size $n \in \{250, 1{,}000\}$. 
The purpose of considering these two sample sizes is as follows: if the theoretical prediction of the root-$n$ convergence rate is correct, then the root mean squared error at $n = 1{,}000$ (denoted by $\mathrm{RMSE}(1000)$) should be approximately one-half of that at $n = 250$ (denoted by $\mathrm{RMSE}(250)$).

As a benchmark, we compute (0) the conventional maximum score estimates. 
Their convergence rate is cube-root-$n$, so $\mathrm{RMSE}(1000)/\mathrm{RMSE}(250)$ will not be approximately one-half, unlike for our surrogate maximum score estimator. 
For the surrogate loss function $\phi$, we consider (1) the logistic loss (with $a=1$), (2) the pseudo-Huber loss (with $a=2$), and (3) the probit loss (with $a=0.5$), as introduced in Section \ref{sec:existence_uniqueness} as examples satisfying our theoretical requirements.
Table \ref{tab:rmse} summarizes the simulated RMSEs. 

%%%%%%%%%%%%%%%%%%%%%%%%%%%%%%%%%%%%%%%%%
\begin{table}[tb]
	\centering
	\caption{Evidence of the Root-$n$ Consistency}\label{tab:rmse}
	\medskip
	\begin{tabular}{llccc}
		\hline
		$X$ & Estimation Method & RMSE($250$) & RMSE($1000$) & $\frac{\mathrm{RMSE}(1000)}{\mathrm{RMSE}(250)}$ \\
		\hline
		\multirow{4}{*}{(i) Normal}
		& (0) Conventional Maximum Score & 0.373 & 0.242 & 0.647 \\
		\cline{2-5}
		& (1) Surrogate Logistic      & 0.202 & 0.101 & \bf 0.502 \\
		& (2) Surrogate Huber         & 0.184 & 0.091 & \bf 0.496 \\
		& (3) Surrogate Probit        & 0.165 & 0.082 & \bf 0.494 \\
		\hline
		\multirow{4}{*}{(ii) $t_5$}
		& (0) Conventional Maximum Score & 0.342 & 0.218 & 0.638 \\
		\cline{2-5}
		& (1) Surrogate Logistic      & 0.176 & 0.088 & \bf 0.499 \\
		& (2) Surrogate Huber         & 0.159 & 0.079 & \bf 0.495 \\
		& (3) Surrogate Probit        & 0.151 & 0.075 & \bf 0.495 \\
		\hline
		\multirow{4}{*}{(iii) Laplace}
		& (0) Conventional Maximum Score & 0.401 & 0.260 & 0.650 \\
		\cline{2-5}
		& (1) Surrogate Logistic      & 0.216 & 0.108 & \bf 0.499 \\
		& (2) Surrogate Huber         & 0.196 & 0.097 & \bf 0.494 \\
		& (3) Surrogate Probit        & 0.181 & 0.090 & \bf 0.494 \\
		\hline
	\end{tabular}
	\begin{minipage}{1\textwidth}
		\small \medskip
		\textit{Notes:} The row groups correspond to three distributions of $X_i$: Normal, $t_5$, and Laplace.
		Within each distribution, the methods are ordered as follows: (0) the conventional maximum score, (1) the surrogate logistic loss (with $a=1$), (2) the surrogate pseudo-Huber loss (with $a=2$), and (3) the surrogate probit loss (with $a=0.5$). 
		The columns report the RMSE for sample sizes $n=250$ and $n=1000$, together with their ratio $\mathrm{RMSE}(1000)/\mathrm{RMSE}(250)$.
		Root-$n$ consistency is evidenced by $\mathrm{RMSE}(1000)/\mathrm{RMSE}(250) \approx 0.5$.
		Each reported value is based on $10{,}000$ Monte Carlo replications.
	\end{minipage}
\end{table}
%%%%%%%%%%%%%%%%%%%%%%%%%%%%%%%%%%%%%%%%%

First, observe that the ratio $\mathrm{RMSE}(1000)/\mathrm{RMSE}(250)$ for (0) the conventional maximum score estimator approximately ranges from 0.64 to 0.65. 
This range is consistent with the well-known fact that the conventional maximum score estimator exhibits a cube-root-$n$ convergence rate. 
Indeed, the theoretically predicted ratio for the cube-root-$n$ consistent estimator is $\mathrm{RMSE}(1000)/\mathrm{RMSE}(250) \approx 0.63$. 

Second, in contrast, all surrogate maximum score estimators (i)--(iii) have RMSE ratios close to $1/2$ (approximately $0.49$--$0.50$) as the sample size increases from $250$ to $1{,}000$, which is consistent with root-$n$ behavior. 
Indeed, the theoretically predicted ratio for the root-$n$ consistent estimator is $\mathrm{RMSE}(1000)/\mathrm{RMSE}(250) = 0.5$. 
These findings align with the theory: the surrogate procedures attain root-$n$ convergence, whereas the original maximum score estimator exhibits the well-known cube-root-$n$ rate.

%%%%%%%%%%%%%%%%%%%%%%%%%%%%%%%%%%%%%%%%%
\subsection{Evidence of the Asymptotic Normality}\label{sec:sim:normal}
%%%%%%%%%%%%%%%%%%%%%%%%%%%%%%%%%%%%%%%%%

While the previous exercise examined the root-$n$ convergence rate, we now turn to the limiting distribution. 
Specifically, our theory predicts that the limiting distribution is normal -- cf. Corollary \ref{cl:root-n}. 
We use simulation studies to demonstrate that the distribution of the surrogate maximum score estimates is indeed well approximated by a normal distribution.

We continue to use the simulation design introduced in Section \ref{sec:design} and focus on the sample size $n=1000$. 
Figure \ref{fig:density} plots the simulated densities of the surrogate maximum score estimates from $10{,}000$ replications as solid lines and overlays them with reference normal densities (constructed using the simulated mean and variance) as dashed lines. 
The first, second, and third rows present results for (1) the surrogate logistic (with $a=1$), (2) he surrogate pseudo-Huber (with $a=2$), and (3) the probit loss (with $a=0.5$), respectively. 
The first, second, and third columns present results under (i) the normal distribution, (ii) the $t_5$ distribution, and (iii) the Laplace distribution for $X_i$, respectively.

\begin{figure}[htbp]
	\centering
	\includegraphics[width=\textwidth]{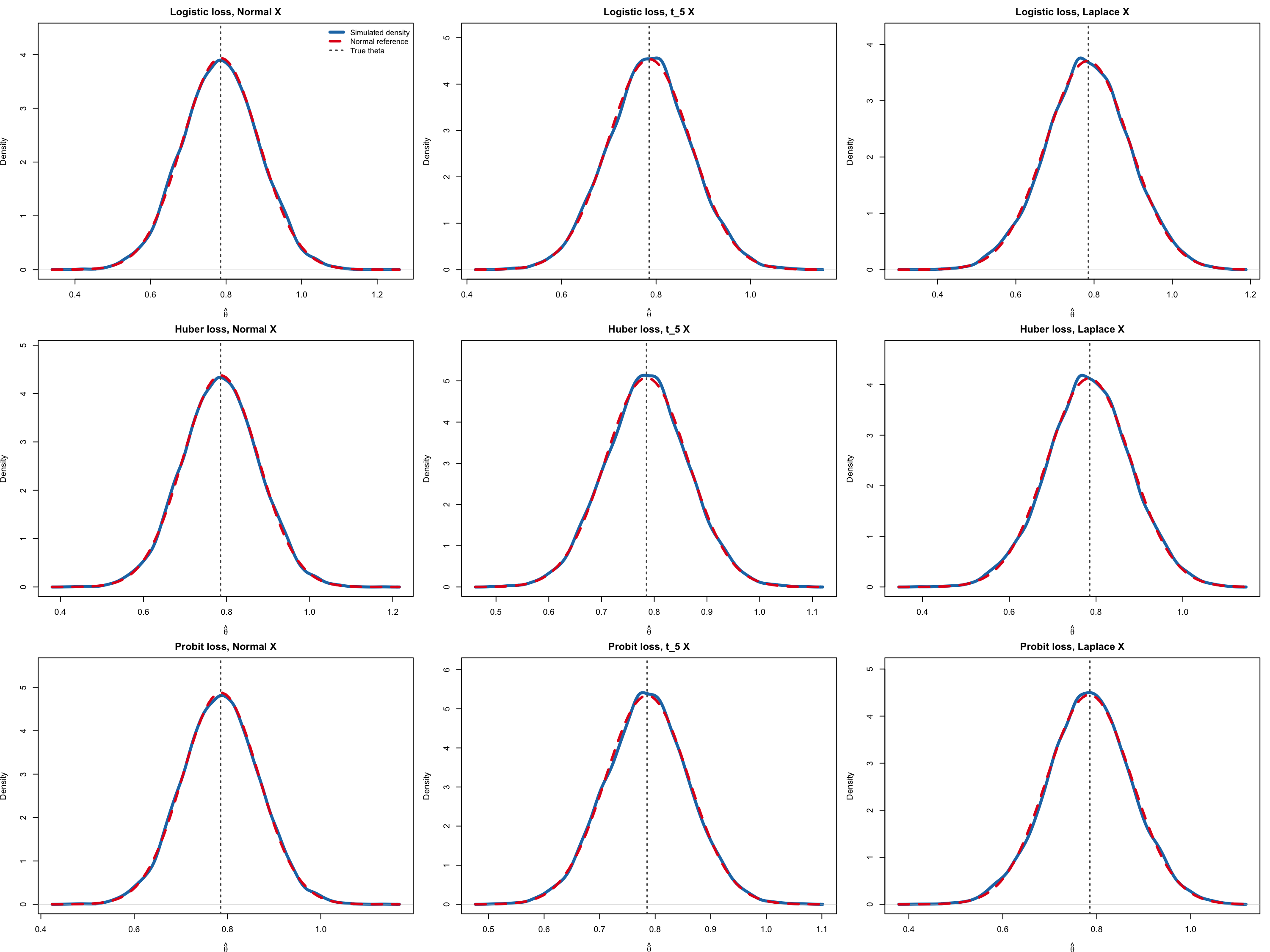}
	\caption{Evidence of the Asymptotic Normality by Density Plots.}\label{fig:density}
	\medskip
	\begin{minipage}{1\textwidth}
		\small
		\textit{Notes:} The first, second, and third rows present results for (1) the surrogate logistic loss (with $a=1$), (2) the surrogate pseudo-Huber loss (with $a=2$), and (3) the surrogate probit loss (with $a=0.5$). respectively. 
		The first, second, and third columns present results under (i) the normal distribution, (ii) the $t_5$ distribution, and (iii) the Laplace distribution for $X_i$, respectively. Solid lines represent the simulated densities of the estimates from $10{,}000$ replications, and red dashed lines represent the corresponding normal reference densities with matched simulation mean and variance. The vertical gray line marks the true parameter value $\theta_0 = \pi/4$.
	\end{minipage}
\end{figure}

Observe that each panel shows the simulated densities closely aligned with the reference normal densities, indicating that the distribution of the surrogate maximum score estimates is well approximated by a normal distribution, consistently with our Corollary \ref{cl:root-n}. 

To further support this observation, Figure \ref{fig:qq} presents the corresponding Q--Q plots for Figure \ref{fig:density}. 
We see that the points lie almost exactly along the 45-degree line, again indicating that the distribution of the surrogate maximum score estimates is well approximated by a normal distribution, supporting our Corollary \ref{cl:root-n}.

\begin{figure}[htbp]
	\centering
	\includegraphics[width=\textwidth]{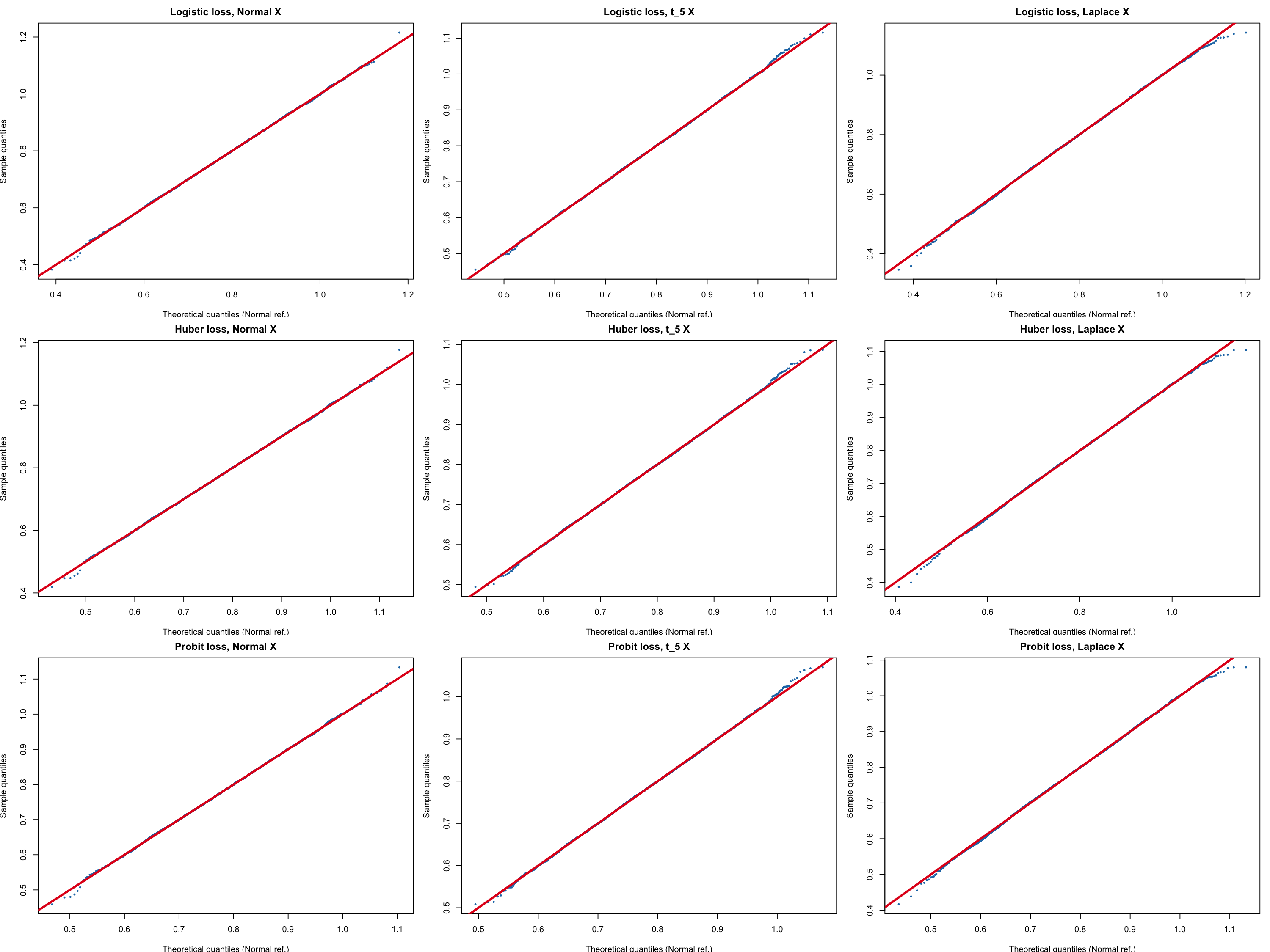}
	\caption{Evidence of the Asymptotic Normality by Q-Q Plots.}\label{fig:qq}
	\medskip
	\begin{minipage}{1\textwidth}
		\small
		\textit{Notes:} The first, second, and third rows present results for (1) the surrogate logistic loss (with $a=1$), (2) the surrogate pseudo-Huber loss (with $a=2$), and (3) the surrogate probit loss (with $a=0.5$). respectively. 
		The first, second, and third columns present results under (i) the normal distribution, (ii) the $t_5$ distribution, and (iii) the Laplace distribution for $X_i$, respectively. 
		Each panel compares the empirical quantiles of the estimators, based on $10{,}000$ Monte Carlo replications, with the corresponding quantiles of a matched normal reference distribution.
	\end{minipage}
\end{figure}

%%%%%%%%%%%%%%%%%%%%%%%%%%%%%%%%%%%%%%%%%
\subsection{Evidence of the Validity of the Standard Inference}\label{sec:sim:inference}
%%%%%%%%%%%%%%%%%%%%%%%%%%%%%%%%%%%%%%%%%
Thus far, we have examined root-$n$ asymptotic normality through simulations to corroborate our theory. 
Given this result, we can conduct standard inference, such as using normal critical values with analytic standard errors (as established in Section \ref{sec:variance}), or bootstrap methods (as established in Section \ref{sec:bootstrap}). 
In this section, we demonstrate the validity of these standard inference procedures through simulations.

We continue to use the simulation design introduced in Section \ref{sec:design} and consider sample sizes $n \in \{250, 1000\}$ to examine how coverage improves as the sample size increases. 
Table \ref{tab:coverage} reports the coverage probabilities of \(95\%\) confidence intervals constructed using two methods: (A) the analytic variance estimator \eqref{eq:var2} together with normal critical values; and (B) the bootstrap procedure. 
The first, second, and third column groups correspond to the (i) normal, (ii) $t_5$, and (iii) Laplace distributions of $X_i$, respectively. 
Within each row group, we report results for (1) the surrogate logistic (with $a=1$), (2) the surrogate pseudo-Huber (with $a=2$), and (3) the surrogate probit (with $a=0.5$) maximum score methods.

%%%%%%%%%%%%%%%%%%%%%%%%%%%%%%%%%%%%%%%%%
\begin{table}[ht]
	\centering
	\caption{Evidence of the Validity of the Standard Inference Methods}\label{tab:coverage}
	\medskip
	\begin{tabular}{llccccc}
		\hline
		\multirow{2}{*}{$X$} & \multirow{2}{*}{Estimation Method} & \multicolumn{2}{c}{(A) Analytic} && \multicolumn{2}{c}{(B) Bootstrap} \\ \cline{3-4}\cline{6-7}
		&  & \(n=250\) & \(n=1000\) && \(n=250\) & \(n=1000\) \\ \hline
		\multirow{3}{*}{(i) Normal}
		& (1) Surrogate Logistic & 0.921 & 0.945 && 0.937 & 0.952 \\
		& (2) Surrogate Huber & 0.931 & 0.945 && 0.941 & 0.949 \\
		& (3) Surrogate Probit & 0.938 & 0.949 && 0.943 & 0.950 \\
		\hline
		\multirow{3}{*}{(ii) $t_5$}
		& (1) Surrogate Logistic & 0.925 & 0.944 && 0.945 & 0.951 \\
		& (2) Surrogate Huber & 0.929 & 0.946 && 0.939 & 0.950 \\
		& (3) Surrogate Probit & 0.935 & 0.949 && 0.943 & 0.951 \\
		\hline
		\multirow{3}{*}{(iii) Laplace}
		& (1) Surrogate Logistic & 0.916 & 0.942 && 0.934 & 0.950 \\
		& (2) Surrogate Huber & 0.929 & 0.948 && 0.940 & 0.953 \\
		& (3) Surrogate Probit & 0.934 & 0.947 && 0.941 & 0.951 \\
		\hline
	\end{tabular}
	\begin{minipage}{1\textwidth}
		\medskip \small
		\textit{Notes:} The row groups correspond to three distributions of $X_i$: Normal, $t_5$, and Laplace.
		Within each distribution, the methods are ordered as follows: (1) the surrogate logistic loss ($a=1$), (2) the surrogate pseudo-Huber loss ($a=2$), and (3) the surrogate probit loss ($a=0.5$).
		The columns report the empirical coverage probabilities of nominal $95\%$ confidence intervals based on (A) the analytic covariance estimator and (B) the bootstrap covariance estimator for sample sizes $n=250$ and $n=1000$. 
		Each reported value is based on $10{,}000$ Monte Carlo replications.
	\end{minipage}
\end{table}
%%%%%%%%%%%%%%%%%%%%%%%%%%%%%%%%%%%%%%%%%

First, we focus on the standard inference based on (A) the analytic variance estimation with normal critical values. 
As the sample size increases from $n=250$ to $n=1000$, the coverage probabilities generally move closer to the nominal level of $0.95$. %This pattern is particularly clear for the logistic and pseudo-Huber surrogates.
These results therefore provide evidence for the validity of the standard inference procedure, as guaranteed by our theory in Sections \ref{sec:asymptotic}--\ref{sec:variance}.

Next, we turn to the standard inference based on (B) the bootstrap method.
Overall, the bootstrap-based confidence intervals deliver coverage probabilities close to the nominal $95\%$ level across all designs. The bootstrap coverage is closer to $0.95$ than its analytic counterpart for for all three surrogate losses when $n=250$ (where the analytic method performed relatively poorly), suggesting the finite-sample improvement.  As the sample size increases to $n=1000$, the analytic and bootstrap procedures perform very similarly. These results suggest that the bootstrap procedure can provide a small-sample refinement, while both methods remain consistent with the asymptotic theory.

%%%%%%%%%%%%%%%%%%%%%%%%%%%%%%%%%%%%%%%%%
\section{Conclusions}\label{sec:conclusion}
%%%%%%%%%%%%%%%%%%%%%%%%%%%%%%%%%%%%%%%%%

The maximum score method is a powerful tool for analyzing semiparametric binary choice models and has been highly influential in the econometrics literature. 
At the same time, it faces a number of practical and theoretical challenges. 
These challenges are part of the reason why the literature has remained active since the method’s introduction. 

In particular, the lack of the standard root-$n$ convergence rate and asymptotic normality has motivated many researchers to develop methods that either address or accommodate these limitations. 
In this paper, we revisit this problem and, rather than addressing or accommodating these challenges, investigate how to achieve root-$n$ asymptotic normality. 

Specifically, we study the conditions under which these standard asymptotic properties hold for surrogate maximum score methods.
These conditions, stated as Conditions (T.\ref{thm:main}.1)--(T.\ref{thm:main}.2) in Theorem \ref{thm:main} are nontrivial.
Yet, they are sufficiently general to encompass a large class of distributions for $X$.

Under these conditions, we establish root-$n$ asymptotic normality and the validity of standard inference methods, such as the use of normal critical values with analytic variance estimation, as well as bootstrap procedures for small-sample refinement. 
Extensive simulation studies corroborate our theoretical claims regarding the root-$n$ convergence rate, asymptotic normality, and the validity of these standard inference methods.

%%%%%%%%%%%%%%%%%%%%%%%%%%%%%%%%%%%%%%%%%
\newpage
\appendix
\section*{Appendix}
%%%%%%%%%%%%%%%%%%%%%%%%%%%%%%%%%%%%%%%%%
\section{Mathematical Proofs}
%%%%%%%%%%%%%%%%%%%%%%%%%%%%%%%%%%%%%%%%%

%%%%%%%%%%%%%%%%%%%%%%%%%%%%%%%%%%%%%%%%%
\subsection{Proof of Lemma \ref{lm:parallel}}\label{sec:lm:parallel}
%%%%%%%%%%%%%%%%%%%%%%%%%%%%%%%%%%%%%%%%%
\bigskip\noindent
\textit{Proof of Lemma \ref{lm:parallel}.}
Let $b_\phi \in \arg\max_{b \in B} Q_\phi(b)$.
Assumption \ref{as:bayes} \eqref{as:bayes:correct} and Condition (L.\ref{lm:parallel}.2) yield
\begin{equation}\label{eq:lm:parallel:equal}
\1\{X'b_\phi \ge 0\} = \1\{X'b_0 \ge 0\} \qquad\text{a.s.}
\end{equation}
We are going to claim that $b_\phi \parallel b_0$.
By way of contradiction, suppose $b_\phi \not\parallel b_0$.
Then,
$$
\Pr(\1\{X'b_\phi \ge 0\} \neq \1\{X'b_0 \ge 0\} ) > 0
$$
by Condition (L.\ref{lm:parallel}.1), contradicting Equation \eqref{eq:lm:parallel:equal}.
Therefore, $b_\phi = c b_0$ for some $c \neq 0$.

We next argue that $c>0$. 
By way of contradiction, suppose $c < 0$.
Then,
$$
\1\{X'b_\phi \ge 0\} = \1\{X'b_0 \leq 0\}
$$
holds.
This equality and Equation \eqref{eq:lm:parallel:equal} imply
$$
\1\{X'b_0 \ge 0\} = \1\{X'b_0 \le 0\} \qquad\text{a.s.}
$$
This almost sure equality implies
$$
\Pr(X'b_0 > 0) = \Pr(X'b_0 < 0) = 0,
$$
contradicting Assumption \ref{as:bayes} \eqref{as:bayes:nondegeneracy}.
Therefore, $c>0$.

Finally, we are going to claim that $b_\phi \in \arg\max_b Q_0 (b)$.
Apply the law of iterated expectations to rewrite $Q_0(b)$ as
\begin{align*}
Q_0(b)
=&
\E[\eta(X) \cdot \1\{X'b \ge 0\} + (1-\eta(X)) \cdot \1\{X'b < 0\}]
\\
=& \E[1 - \eta(X)] + 2 \E[(\eta(X)-1/2) \cdot \1\{X'b \ge 0\}].
\end{align*}
The last term in this equation and Assumption \ref{as:bayes} \eqref{as:bayes:correct} imply that $b_0 \in \arg\max_{b \in B} Q_0(b)$.
Thus, Equation \eqref{eq:lm:parallel:equal} further implies that $b_\phi \in \arg\max_{b \in B} Q_0(b)$.
\qed
%%%%%%%%%%%%%%%%%%%%%%%%%%%%%%%%%%%%%%%%%

%%%%%%%%%%%%%%%%%%%%%%%%%%%%%%%%%%%%%%%%%
\subsection{Proof of Lemma \ref{lm:existence_uniqueness}}\label{sec:lm:existence_uniqueness}
%%%%%%%%%%%%%%%%%%%%%%%%%%%%%%%%%%%%%%%%%
\bigskip\noindent
\textit{Proof of Lemma \ref{lm:existence_uniqueness}.}
(Existence)
We can write $Q_\phi(b) = \E[f((Y,X);b)]$, where
$$
f((y,x);b) = y \cdot \phi(x'b) + (1-y) \cdot \phi(-x'b).
$$
By Assumption \ref{as:moments} \eqref{as:moments:bounded}, the function $\overline f$ defined by
$$
\overline f(y,x) = \sup_{b \in B} (|\phi(x'b)| + |\phi(-x'b)|),
$$
which does not depend on $b$, serves as an integrable envelop of $f$ uniformly on $B$.
Let $b_n \to b$ in $B$. 
Then, by the continuity of $\phi$ guaranteed under Assumption \ref{as:surrogate}, we have $f((y,x);b_n) \to f((y,x);b)$ for each $(y,x)$.
With the aforementioned uniform integrability, therefore, Lebesgue Dominated Convergence Theorem yields $Q_\phi(b_n) \to Q_\phi(b)$, showing that $Q_\phi$ is continuous on $B$.
Assumption \ref{as:bayes} requires that $B$ is compact and nonempty.
Therefore, Weierstrass Theorem concludes that $\arg\max_{b \in B} Q_\phi(b)$ is nonempty.

\bigskip\noindent
(Uniqueness)
Let $b_1, b_2 \in B$ with $b_1 \neq b_2$, and let $\lambda \in (0,1)$.
Then,
$$
b_\lambda = \lambda b_1 + (1-\lambda) b_2 \in B
$$
by Assumption \ref{as:moments} \eqref{as:moments:convex}.
Because $\phi$ is concave under Assumption \ref{as:surrogate},
\begin{align}
\phi(\lambda X'b_1 + (1-\lambda) X'b_2) \ge& \lambda \cdot \phi(X'b_1) + (1-\lambda) \cdot \phi(X'b_2)
\label{eq:lm:existence_uniqueness:concave1}
\\
\phi(-\lambda X'b_1 - (1-\lambda) X'b_2) =& \phi(\lambda(-X'b_1) + (1-\lambda)(-X'b_2)
\notag
\\
\ge& \lambda \phi(-X'b_1) + (1-\lambda) \phi(-X'b_2)
\label{eq:lm:existence_uniqueness:concave2}
\end{align}
Adding $Y$ times \eqref{eq:lm:existence_uniqueness:concave1} and $(1-Y)$ times \eqref{eq:lm:existence_uniqueness:concave2}, and then taking the expectation yields
$$
Q_\phi(b_\lambda) \geq \lambda Q_\phi(b_1) + (1-\lambda) Q_\phi(b_2),
$$
showing that $Q_\phi$ is concave on $B$.

To show the strict concavity, note that the inequalities \eqref{eq:lm:existence_uniqueness:concave1}--\eqref{eq:lm:existence_uniqueness:concave2} hold with strict inequalities whenever $X'(b_1-b_2) \neq 0$.
Thus, $Y$ times \eqref{eq:lm:existence_uniqueness:concave1} entails a strict inequality on the event $A_1$ where $X'(b_1-b_2) \neq 0$ and $Y=1$.
Likewise, $(1-Y)$ times \eqref{eq:lm:existence_uniqueness:concave2} entails a strict inequality on the event $A_2$ where $X'(b_1-b_2) \neq 0$ and $Y=0$.
Since $Y$ is $\{0,1\}$-valued, Assumption \ref{as:moments} \eqref{as:moments:index} implies $\Pr(A_1 \cup A_2) > 0$.
Therefore, $\Pr(A_1) > 0$ or $\Pr(A_2) > 0$, and the same addition and integration as above yields
$$
Q_\phi(b_\lambda) > \lambda Q_\phi(b_1) + (1-\lambda) Q_\phi(b_2),
$$
showing that $Q_\phi$ is strictly concave on $B$.

\bigskip
The convexity of $B$ and the strict concavity of $Q_\phi$ as shown above implies the uniqueness of the solution to $\max_{b \in B} Q_\phi(b)$.
\qed
%%%%%%%%%%%%%%%%%%%%%%%%%%%%%%%%%%%%%%%%%

%%%%%%%%%%%%%%%%%%%%%%%%%%%%%%%%%%%%%%%%%
\subsection{Proof of Proposition \ref{pp:nonzero_probability}}\label{sec:pp:nonzero_probability}
%%%%%%%%%%%%%%%%%%%%%%%%%%%%%%%%%%%%%%%%%
\bigskip\noindent
\textit{Proof of Proposition \ref{pp:nonzero_probability}.}
Let $b_1, b_2 \in B$ be nonparallel.
Define
$$
D = \{x \in \mathbb{R}^d : \1\{x'b_1 \ge 0\} \neq \1\{x'b_2 \ge 0\}\}.
$$
Let $r \in (0,\infty)$.
We are going to show that $\lambda_d( D \cap B_r(0) ) > 0$. 
Let 
$$
v = b_1 - \frac{b_1'b_2}{\|b_2\|^2} b_2.
$$
Then, $v'b_2 = 0$ and also $v \neq 0$ since $b_1 \not\parallel b_2$.
Furthermore,
$$
v'b_1 = \|b_1\|^2 - \frac{(b_1'b_2)^2}{\|b_2\|^2} > 0
$$
by Cauchy-Schwarz inequality, where the strictly inequality is due to $b_1 \not\parallel b_2$.

\bigskip\noindent
{\bf Case: $b_1'b_2 \le 0$.}
In this case, $x_\varepsilon = -b_2 + \varepsilon v$ for $\varepsilon > 0$ satisfies
$$
x_\varepsilon'b_2 = -\|b_2\|^2 < 0
\qquad\text{and}\qquad
x_\varepsilon'b_1 = -b_2'b_1 + \varepsilon v'b_1 > 0.
$$

\medskip\noindent
{\bf Case: $b_1'b_2 > 0$.}
In this case, $x_\varepsilon = b_2 - \varepsilon v$ for $\varepsilon > (b_2'b_1)/(v'b_1)$ satisfies
$$
x_\varepsilon'b_2 = \|b_2\|^2 > 0
\qquad\text{and}\qquad
x_\varepsilon'b_1 = b_2'b_1 - \varepsilon v'b_1 < 0.
$$

\bigskip\noindent
These two cases put together, at least one of the two sets
$$
U_{12} = \{x \in \mathbb{R}^d : x'b_1 > 0 \text{ and } x'b_2 < 0\}
\qquad\text{and}\qquad
U_{21} = \{x \in \mathbb{R}^d : x'b_1 < 0 \text{ and } x'b_2 > 0\}
$$
is nonempty.
Since $U_{12} \cup U_{21}$ is open as a union of open sets and $\emptyset \neq U_{12} \cup U_{21} \subseteq D$, there exists a nonempty open set $D^o \subset D$.
Because $D$ is a cone, there exists sufficiently small $\delta \in (0,\infty)$ such that $r D^o \in B_r(0)$.
But then,
$$
\lambda_d(D \cap B_r(0)) \ge \lambda_d(\delta D^o) > 0.
$$

Now apply Assumption \ref{as:local_full_support} with $A = D \cap B_r(0)$ to conclude
$$
\Pr(\1\{X'b_1 \ge 0\} \neq \1\{X'b_2 \ge 0\})
\ge
\Pr(X \in D \cap B_r(0)) > 0
$$
as claimed in the statement of the proposition.
\qed
%%%%%%%%%%%%%%%%%%%%%%%%%%%%%%%%%%%%%%%%%

%%%%%%%%%%%%%%%%%%%%%%%%%%%%%%%%%%%%%%%%%
\subsection{Proof of Proposition \ref{pp:surrogate}}\label{sec:pp:surrogate}
%%%%%%%%%%%%%%%%%%%%%%%%%%%%%%%%%%%%%%%%%
\bigskip\noindent
\textit{Proof of Proposition \ref{pp:surrogate}.}
Let $b_\phi = \arg\max_{b \in B} Q_\phi(b)$, which is unique under Assumptions \ref{as:bayes}, \ref{as:surrogate}, and \ref{as:moments} by Lemma \ref{lm:existence_uniqueness}.
We are going to show that it satisfies $\1\{X'b_\phi \ge 0\} = \1\{\eta(X) \ge 1/2\}$ a.s.
Define a function $r: \mathbb{R}^2 \to \mathbb{R}$ by
$$
r(t,u) = h(t) \cdot \phi(u) + (1-h(t)) \cdot \phi(-u)
$$
By Assumption \ref{as:index} \eqref{as:index:single_index}, we can write $Q_\phi(b)$ in terms of this function as
$$
Q_\phi(b) = \E[r(T,X'b)].
$$
Let $b \in \mathbb{R}^d$.
By Assumption \ref{as:index} \eqref{as:index:linear}, there exists $a_b \in \mathbb{R}$ and a random variable $W_b$ such that 
$$
X'b = a_b T + W_b, \qquad\E[W_b|T]=0.
$$
Hence, $Q_\phi(b)$ can be further rewritten as
$$
Q_\phi(b) = \E[r(T,a_b T + W_b)].
$$
Since $\phi$ is strictly concave under Assumption \ref{as:surrogate}, the map
$$
\mathbb{R} \ni u \mapsto r(t,u) = h(t)\phi(u) + (1-h(t))\phi(-u)
$$
is strictly concave for each $t \in \mathbb{R}$.
Therefore, we can apply the conditional Jensen's inequality to obtain
$$
\E[r(T,a_b T + W_b)|T] 
\le
r(T,\E[a_b T + W_b|T])
=
r(T,a_b T),
$$
where the equality follows from $\E[W_b|T]=0$ and the conditioning theorem.
Taking expectations on both sides and applying the law of iterated expectations, we get
$$
Q_\phi(b) \le Q_\phi(a_b b_0).
$$
In particular, since $b_\phi \in \arg\max_{b \in B} Q_\phi(b)$, it follows that
$$
Q_\phi(b_\phi) \leq Q_\phi(a_{b_\phi} b_0) \leq Q_\phi(b_\phi),
$$
and the uniqueness $\{b_\phi\} = \arg\max_{b \in B} Q_\phi(b)$ further implies $b_\phi = a_{b_\phi} b_0$.
Define 
$$
\Psi(a) = Q_\phi(ab_0) = \E[r(T,aT)].
$$
We are going to argue that $a_{b_\phi} > 0$.
By way of contradiction, suppose that $a_{b_\phi} \leq 0$ maximizes $\Psi$.
Let us branch into two cases.

\bigskip\noindent
{\bf Case:} $a_{b_\phi} < 0$.
If $t>0$, then $h(t) > 1/2$ by Assumption \ref{as:index} \eqref{as:index:single_index}.
Since $\phi$ is strictly increasing under Assumption \ref{as:surrogate}, it follows that
\begin{align*}
&r(t,|a_{b_\phi}|t)-r(t,a_{b_\phi} t)
\\
=&
h(t) \cdot \phi(|a_{b_\phi}|t) + (1-h(t)) \cdot \phi(-|a_{b_\phi}|t)
-
h(t) \cdot \phi(-|a_{b_\phi}|t) - (1-h(t)) \cdot \phi(|a_{b_\phi}|t)
\\
=&
(2h(t) - 1) \cdot (\phi(|a_{b_\phi}|t) - \phi(-|a_{b_\phi}|t))
>0.
\end{align*}
Similar lines of calculations yield
$
r(t,|a_{b_\phi}|t) - r(t,a_{b_\phi} t) > 0
$
for the case of $t<0$ too.
Thus,
$
r(t,|a_{b_\phi}|t) - r(t,a_{b_\phi} t) > 0
$
holds whenever $t \neq 0$.
Since Assumption \ref{as:bayes} \eqref{as:bayes:nondegeneracy} implies $\Pr(T \neq 0)>0$,
$$
\Psi(|a_{b_\phi}|) > \Psi(a_{b_\phi}).
$$
Note that Assumption \ref{as:moments} \eqref{as:moments:convex} implies that $ab_0 \in B \Longrightarrow |a|b_0 \in B$. 
But then, the above inequality is a contradiction with $a_{b_\phi}$ being a maximizer of $\Psi$.
$\blacktriangle$

\bigskip\noindent
{\bf Case:} $a_{b_\phi} = 0$.
By Assumption \ref{as:surrogate}, we can write the derivative of $\Psi$ as
$$
\Psi'(0) = \phi'(0) \E[T(2h(T)-1)],
$$
where Assumptions \ref{as:surrogate} and \ref{as:index} \eqref{as:index:bounded_moment}--\eqref{as:index:single_index} ensures the differentiability.
Assumption \ref{as:index} \eqref{as:index:single_index} implies
$$
T (2h(T)-1) \geq 0 \qquad\text{a.s.}
$$
with the inequality satisfied with strictly inequality when $T \neq 0$.
Since Assumption \ref{as:bayes} \eqref{as:bayes:nondegeneracy} implies $\Pr(T \neq 0)>0$, integration yields
$$
\E[T(2h(T)-1)]>0.
$$
Since $\phi'(0)>0$ under Assumption \ref{as:surrogate}, it follows that 
$$
\Psi'(0)>0,
$$
contradicting with the first-order condition for $a_{b_\phi} = 0$ being the maximizer of $\Psi$.
$\blacktriangle$

\bigskip\noindent
Since we have shown that $a_{b_\phi}>0$, it follows that
$$
\1\{X'b_\phi \ge 0\} = \1\{a_{b_\phi} T \ge 0\} = \1\{T \ge 0\}.
$$
Furthermore, Assumption \ref{as:bayes} \eqref{as:bayes:correct} implies
$$
\1\{T \ge 0\} = \1\{\eta(X) \ge 1/2\} \qquad\text{a.s.}
$$
Combining the last two inequalities yields
$
\1\{X'b_\phi \ge 0\} = \1\{\eta(X) \ge 1/2\}
$
almost surely.
\qed
%%%%%%%%%%%%%%%%%%%%%%%%%%%%%%%%%%%%%%%%%

%%%%%%%%%%%%%%%%%%%%%%%%%%%%%%%%%%%%%%%%%
\subsection{Proof of Corollary \ref{cl:root-n}}\label{sec:thm_root-n}
%%%%%%%%%%%%%%%%%%%%%%%%%%%%%%%%%%%%%%%%%
\bigskip\noindent
\textit{Proof of Corollary \ref{cl:root-n}}. To show consistency of $\widehat{b}$, we verify the conditions of the argmax theorem \citep[Theorem 2.1]{newey1994large}. Note that
\begin{align*}
Q_{\phi, n}(b)=\frac{1}{n} \sum_{i=1}^n \ell_\phi(Z_i, b),~~~~ Q_\phi(b)=\mathbb{E}[\ell_\phi(Z_i, b)],
\end{align*}
where $\ell_\phi(Z_i, b):=Y_i \phi(X_i^{\prime} b)+(1-Y_i) \phi(-X_i^{\prime} b)$. 

We first show the uniform convergence
\begin{align}\label{eq:cond1}
\sup _{b \in B}\left\vert Q_{\phi, n}(b)-Q_\phi(b)\right\vert \rightarrow_{p} 0.
\end{align}
Note that $\mathbb{E}[F_0(Z_i)]\leq\mathbb{E}[\vert F_0(Z_i)\vert ]\leq (\mathbb{E}[F_0(Z_i)^2])^{1/2}<\infty$ by Assumption \ref{as:asymptotic} \eqref{as:asymptotic4}.
Also, $z \mapsto \ell_\phi(z, b)$ is continuous for each $b \in B$ (Assumption \ref{as:asymptotic} \eqref{as:asymptotic3}). 
Thus, the function class $\{\ell_\phi(\cdot, b): b \in B\}$ is pointwise measurable, uniformly bounded by the integrable envelope $F_0$ and equicontinuous in $b$ under Assumption \ref{as:asymptotic} \eqref{as:asymptotic3}--\ref{as:asymptotic} \eqref{as:asymptotic4}. In addition, the parameter space $B$ is compact (Assumption \ref{as:asymptotic} \eqref{as:asymptotic0}). 
Hence, by the Uniform Weak Law of Large Numbers \citep[][Lemma 2.4]{newey1994large}, we obtain \eqref{eq:cond1} where $Q_\phi(b)$ is continuous.

By Assumption \ref{as:asymptotic} \eqref{as:asymptotic3}, $Q_\phi$ is twice continuously differentiable near $b_\phi$ and $H=\nabla^2_b Q_\phi(b_\phi)$ is negative definite.
Therefore, by the definition of $b_\phi$,
Taylor's theorem yields
\begin{align*}
Q_\phi(b)=Q_\phi(b_\phi)+\frac{1}{2}(b-b_\phi)^{\prime} H(b-b_\phi)+o(\Vert b-b_\phi\Vert^2).
\end{align*}
Since $H$ is negative definite, there exists $c>0$ such that $(b-b_\phi)^{\prime} H(b-b_\phi) \leq-c\Vert b-b_\phi\Vert^2$. 
Thus,
\begin{align}\label{eq:unique}
Q_\phi(b) \leq Q_\phi(b_\phi)-c\Vert b-b_\phi\Vert^2
\end{align}
for all $b$ in a neighborhood of $b_\phi$. 
By the compactness of $B$ (Assumption \ref{as:asymptotic} \eqref{as:asymptotic0}), the uniform convergence \eqref{eq:cond1}, the fact $Q_\phi$ is continuous, and the inequality \eqref{eq:unique}, the argmax continuous mapping theorem \citep[][Theorem 2.1]{newey1994large} implies the consistency of $\widehat{b}$.

Next, we are going to show the asymptotic normality. 
By Assumption \ref{as:asymptotic} \eqref{as:asymptotic3}, we can apply the coordinate-wise mean value expansion to obtain
\begin{align}\label{eq:mean_value}
0=\sqrt{n} \nabla_b Q_{\phi, n}(b_\phi)+\nabla_b^2 Q_{\phi, n}(\widetilde{b}) \sqrt{n}(\widehat{b}-b_\phi),
\end{align}
where $\widetilde{b}$ lies on the line segment between $\widehat{b}$ and $b_\phi$.

By the definition of $b_\phi$ and Assumption \ref{as:asymptotic} \eqref{as:asymptotic3}, we have
\begin{align*}
\mathbb{E}[\psi(Z_i, b_\phi)]=\nabla_b Q_\phi(b_\phi)=\mathbf{0}_{d\times 1} 
\quad\text{and}\quad
\Omega=\mathbb{E}[\psi(Z_i, b_\phi) \psi(Z_i, b_\phi)^{\prime}]<\infty.
\end{align*}
Since $\{Z_i\}$ are i.i.d. under Assumption \ref{as:asymptotic} \eqref{as:asymptotic1}, it follows from the multivariate Lindeberg-Feller Central Limit Theorem that
\begin{align}\label{eq1}
\sqrt{n} \nabla_b Q_{\phi, n}
=
\frac{1}{\sqrt{n}} \sum_{i=1}^n \psi(Z_i, b_\phi) \rightarrow_{d} \mathcal{N}(\mathbf{0}_{d\times 1}, \Omega).
\end{align}

The map $z \mapsto \nabla_b^2 \ell_\phi(z, b)$ is continuous for each $b \in B$ by Assumption \ref{as:asymptotic} \eqref{as:asymptotic3}.
Therefore, $\{\nabla_b^2 \ell_\phi(\cdot,b) : b \in B\}$ is pointwise measurable.
Note also that its envelope satisfies $\mathbb{E}[F_2(Z_i)]\leq \mathbb{E}[\vert F_2(Z_i)\vert]\leq (\mathbb{E}[F_2(Z_i)^2])^{1/2}<\infty$ by Assumption \ref{as:asymptotic} \eqref{as:asymptotic4}.
Furthermore, the parameter space $B$ is compact (Assumption \ref{as:asymptotic} \eqref{as:asymptotic0}).
By the Uniform Weak Law of Large Numbers \citep[][Lemma 2.4]{newey1994large},  therefore,
\begin{equation}\label{eq:ulln_nabla2}
\sup _{b \in B}\Vert \nabla_b^2 Q_{\phi,n}(b)-\nabla_b^2 Q_\phi(b)\Vert \rightarrow_{p} 0.
\end{equation}
In turn, this uniform consistency \eqref{eq:ulln_nabla2}, the triangle inequality, the fact that $\widetilde{b}$ lies on the line segment between $\widehat{b}$ and $b_\phi$, and the consistency of $\widehat{b}$ established in the first part of the proof together yield
\begin{equation}\label{eq:consistency_H}
\nabla_b^2 Q_{\phi,n}(\widetilde{b}) \rightarrow_{p} H.
\end{equation}

Combining \eqref{eq:mean_value}, \eqref{eq1} and \eqref{eq:consistency_H} yields the desired asymptotic normality result.
\qed

%%%%%%%%%%%%%%%%%%%%%%%%%%%%%%%%%%%%%%%%%
\subsection{Proof of Corollary \ref{cl:var}}\label{sec:thm_var}
%%%%%%%%%%%%%%%%%%%%%%%%%%%%%%%%%%%%%%%%%

\bigskip\noindent\textit{Proof of Corollary \ref{cl:var}}.
We are going to show the consistency of $\widehat{H}$ and $\widehat{\Omega}$ for $H$ and $\Omega$, respectively.
Consequently, the statement follows by the Continuous Mapping Theorem under Assumption \ref{as:asymptotic} \eqref{as:asymptotic3}.

\medskip
\noindent\textbf{Step 1: Consistency of $\widehat{H}$.}
Consider the decomposition
\begin{equation}\label{eq:decomposition_H}
\widehat{H} - H
=
B_1 + B_2,
\end{equation}
where
\begin{align*}
B_1 =& {\nabla_b^2 Q_{\phi,n}(\widehat{b}) - \nabla_b^2 Q_{\phi,n}(b_\phi)}
\quad\text{and}
\\
B_2 =& {\nabla_b^2 Q_{\phi,n}(b_\phi) - H}.
\end{align*}

For the second term $B_2$ in the decomposition \eqref{eq:decomposition_H}, we have
$
B_2 \rightarrow_p 0
$
by \eqref{eq:ulln_nabla2}.
For the first term $B_1$ in the decomposition \eqref{eq:decomposition_H}, we have
$
B_2 \rightarrow_p 0
$
by the consistency of $\widehat{b}$ and the continuity of $\nabla_b^2 Q_{\phi,n}$ under Assumption \ref{as:asymptotic} \eqref{as:asymptotic3}.

\medskip
\noindent\textbf{Step 2: Consistency of $\widehat{\Omega}$.} 
Consider the decomposition
\begin{equation}\label{eq:decomposition_Omega}
\widehat{\Omega} - \Omega
= 
C_1 + C_2,
\end{equation}
where
\begin{align*}
C_1 =& {\widehat{\Omega} - \frac{1}{n}\sum_{i=1}^n \psi(Z_i,b_\phi)\psi(Z_i,b_\phi)'}
\quad\text{and}
\\
C_2 =& {\frac{1}{n}\sum_{i=1}^n \psi(Z_i,b_\phi)\psi(Z_i,b_\phi)' - \Omega}.
\end{align*}
For the second term $C_2$ in \eqref{eq:decomposition_Omega}, we have
\begin{align*}
C_2 \rightarrow_{p}0,
\end{align*}
by the Law of Large Numbers under Assumption \ref{as:asymptotic} \eqref{as:asymptotic1} and the integrability condition on the summand which can be shown by Assumption \ref{as:asymptotic} \eqref{as:asymptotic3} that controls the effect of $\Omega$ through its bounded eigenvalues.

Now it remains to show that the first term \(C_1\) in the decomposition \eqref{eq:decomposition_Omega} also converges in probability to zero.
By the coordinate-wise Mean Value Theorem, for each \(z\), there exists some \(\widetilde b(z)\) on the
	line segment between \(\widehat b\) and \(b_\phi\) such that
	\[
	\psi(z,\widehat b)-\psi(z,b_\phi)
	=
	\nabla_b^2\ell_\phi(z,\widetilde b(z))\,(\widehat b-b_\phi).
	\]
	Hence,
	\begin{align}
		\label{eq:bound_psi_revised2}
		\Vert \psi(z,\widehat b)-\psi(z,b_\phi)\Vert
		&\le
		\Vert \nabla_b^2\ell_\phi(z,\widetilde b(z))\Vert\,\Vert \widehat b-b_\phi\Vert
		\le
		F_2(z)\Vert \widehat b-b_\phi\Vert .
	\end{align}
	Moreover, by definition of \(F_1\) of Assumption \ref{as:asymptotic} \eqref{as:asymptotic4},
	\begin{align}
		\label{eq:bound_psi_bphi_revised2}
		\Vert \psi(z,b_\phi)\Vert
		=
		\Vert \nabla_b\ell_\phi(z,b_\phi)\Vert
		\le
		F_1(z).
	\end{align}
	Therefore, by the triangle inequality and \eqref{eq:bound_psi_revised2},
	\begin{align}
		\label{eq:bound_psi_hat_revised2}
		\Vert \psi(z,\widehat b)\Vert
		&\le
		\Vert \psi(z,b_\phi)\Vert+\Vert \psi(z,\widehat b)-\psi(z,b_\phi)\Vert
		\le
		F_1(z)+F_2(z)\Vert \widehat b-b_\phi\Vert.
	\end{align}
Add and subtract \(\psi(z,b_\phi)\psi(z,\widehat b)^\prime\) to obtain
\begin{align*}
		&\psi(z,\widehat b)\psi(z,\widehat b)^\prime
		-
		\psi(z,b_\phi)\psi(z,b_\phi)^\prime=
		\bigl(\psi(z,\widehat b)-\psi(z,b_\phi)\bigr)\psi(z,\widehat b)^\prime
		+
		\psi(z,b_\phi)\bigl(\psi(z,\widehat b)-\psi(z,b_\phi)\bigr)^\prime.
\end{align*}
	Thus, by the triangle inequality and submultiplicativity of the matrix norm,
	\begin{align*}
		&\left\Vert
		\psi(z,\widehat b)\psi(z,\widehat b)^\prime
		-
		\psi(z,b_\phi)\psi(z,b_\phi)^\prime
		\right\Vert \\
		&\qquad\le
		\Vert \psi(z,\widehat b)-\psi(z,b_\phi)\Vert\,\Vert \psi(z,\widehat b)\Vert
		+
		\Vert \psi(z,b_\phi)\Vert\,\Vert \psi(z,\widehat b)-\psi(z,b_\phi)\Vert \\
		&\qquad=
		\Vert \psi(z,\widehat b)-\psi(z,b_\phi)\Vert
		\Bigl(\Vert \psi(z,\widehat b)\Vert+\Vert \psi(z,b_\phi)\Vert\Bigr).
	\end{align*}
Applying \eqref{eq:bound_psi_revised2}--\eqref{eq:bound_psi_hat_revised2}, we obtain
	\begin{align*}
		&\left\Vert
		\psi(z,\widehat b)\psi(z,\widehat b)^\prime
		-
		\psi(z,b_\phi)\psi(z,b_\phi)^\prime
		\right\Vert \\
		&\qquad\le
		F_2(z)\Vert \widehat b-b_\phi\Vert
		\Bigl(
		F_1(z)+F_2(z)\Vert \widehat b-b_\phi\Vert+F_1(z)
		\Bigr) \\
		&\qquad\lesssim
		F_2(z)\Vert \widehat b-b_\phi\Vert
		\Bigl(
		F_1(z)+F_2(z)\Vert \widehat b-b_\phi\Vert
		\Bigr).
	\end{align*}
	Consequently,
	\begin{align}
		&\left\Vert
		\frac{1}{n}\sum_{i=1}^n
		\Bigl(
		\psi(Z_i,\widehat b)\psi(Z_i,\widehat b)^\prime
		-
		\psi(Z_i,b_\phi)\psi(Z_i,b_\phi)^\prime
		\Bigr)
		\right\Vert \notag\\
		&\le
		\frac{1}{n}\sum_{i=1}^n
		\left\Vert
		\psi(Z_i,\widehat b)\psi(Z_i,\widehat b)^\prime
		-
		\psi(Z_i,b_\phi)\psi(Z_i,b_\phi)^\prime
		\right\Vert \notag\\
		&\lesssim
		\Vert \widehat b-b_\phi\Vert
		\frac{1}{n}\sum_{i=1}^n
		F_2(Z_i)
		\Bigl(
		F_1(Z_i)+F_2(Z_i)\Vert \widehat b-b_\phi\Vert
		\Bigr) \notag\\
		&=
		\Vert \widehat b-b_\phi\Vert
		\left(
		\frac{1}{n}\sum_{i=1}^n F_1(Z_i)F_2(Z_i)
		\right)
		+
		\Vert \widehat b-b_\phi\Vert^2
		\left(
		\frac{1}{n}\sum_{i=1}^n F_2(Z_i)^2
		\label{eq:C1_bound_revised2}
		\right).
	\end{align}
	
By Cauchy--Schwarz inequality,
	\[
	\mathbb E[F_1(Z)F_2(Z)]
	\le
	\bigl(\mathbb E[F_1(Z)^2]\bigr)^{1/2}
	\bigl(\mathbb E[F_2(Z)^2]\bigr)^{1/2}
	<\infty
	\]
under Assumption \ref{as:asymptotic} \eqref{as:asymptotic4}.
Therefore, the Law of Large Numbers assisted under Assumption \ref{as:asymptotic} \eqref{as:asymptotic1} yields
	\[
	\frac{1}{n}\sum_{i=1}^n F_1(Z_i)F_2(Z_i)=O_p(1),
	\qquad
	\frac{1}{n}\sum_{i=1}^n F_2(Z_i)^2=O_p(1).
	\]
Since \(\widehat b\rightarrow_{p} b_\phi\) by Corollary \ref{cl:root-n}, we have \(\Vert \widehat b-b_\phi\Vert=o_p(1)\). It then follows from
	\eqref{eq:C1_bound_revised2} that
	\[
	\left\Vert
	\frac{1}{n}\sum_{i=1}^n
	\Bigl(
	\psi(Z_i,\widehat b)\psi(Z_i,\widehat b)^\prime
	-
	\psi(Z_i,b_\phi)\psi(Z_i,b_\phi)^\prime
	\Bigr)
	\right\Vert
	=o_p(1),
	\]
	which establishes \(C_1\rightarrow_{p}0\).

%\newpage
%By the law of large numbers, Assumption \ref{as:asymptotic}.\eqref{as:asymptotic4} and Cauchy-Schwarz inequality, 
%\begin{align*}
%\frac{1}{n}\sum_{i=1}^n F_2(Z_i)F_1(Z_i)=O_p(1)
%\quad\text{and}\quad
%\frac{1}{n}\sum_{i=1}^n F_2(Z_i)^2=O_p(1).
%\end{align*}
%Since $\widehat{b} \rightarrow_{p} b_\phi$ by Corollary \ref{cl:root-n}, it follows that $C_2=o_p(1)$.

%Therefore, $\widehat{\Omega}\rightarrow_{p}\Omega$.

\medskip
\noindent\textbf{Step 3: Consistency of $\widehat{V}$.}
In light of the results from Steps 1 and 2, the Continuous Mapping Theorem yields
$
\widehat{V}=\widehat{H}^{-1}\widehat{\Omega}\widehat{H}^{-1}\rightarrow_{p}H^{-1}\Omega H^{-1}$.
\qed

%%%%%%%%%%%%%%%%%%%%%%%%%%%%%%%%%%%%%%%%%
\subsection{Proof of Corllary \ref{cl:studentized_bootstrap}}\label{sec:studentized_bootstrap}
%%%%%%%%%%%%%%%%%%%%%%%%%%%%%%%%%%%%%%%%%
\noindent\textit{Proof of Corllary \ref{cl:studentized_bootstrap}}. 
The conditions of Corollary \ref{cl:root-n} yield the linear expansion
\begin{align}
	\label{eq:lin_rep_revised}
	\widehat b-b_\phi
	=
	H^{-1}\frac{1}{n}\sum_{i=1}^n \psi(Z_i,b_\phi)+r_n,~~\text{with}~~
	\Vert r_n\Vert=o_p(n^{-1/2}).
\end{align}
Since \(\theta=a^\prime b_\phi\) and \(\widehat\theta=a^\prime \widehat b\), multiplying
\eqref{eq:lin_rep_revised} by \(a^\prime\) gives
\begin{align}
	\label{eq:theta_lin_revised}
	\sqrt n(\widehat\theta-\theta)
	=
	\frac{1}{\sqrt n}\sum_{i=1}^n \varphi(Z_i)+o_p(1),~~\text{with}~~\varphi(Z):=a^\prime H^{-1}\psi(Z,b_\phi).
\end{align}
Because \(\mathbb E[\psi(Z,b_\phi)]=0\) and
\(\mathbb E[\psi(Z,b_\phi)\psi(Z,b_\phi)^\prime]=\Omega\), we have
\begin{align*}
	\mathbb E[\varphi(Z)]=0~~\text{and}~~\mathbb E[\varphi(Z)^2]
	=
	a^\prime H^{-1}\Omega H^{-1}a
	=:\sigma^2.
\end{align*}
Then \eqref{eq:theta_lin_revised} implies
\begin{align}
	\label{eq:oracle_equiv}
	\frac{\sqrt n(\widehat\theta-\theta)}{\sigma}
	=
	\frac{1}{\sigma\sqrt n}\sum_{i=1}^n \varphi(Z_i)+o_p(1).
\end{align}
It shows that the leading term of the centered estimator is the normalized sum of i.i.d. variables \(\varphi(Z_1),\ldots,\varphi(Z_n)\). Hence, under Assumptions 
\ref{as:asymptotic}\eqref{as:asymptotic1} and \ref{as:bootstrap}, Edgeworth Theorem \citep[Equation (5.5) and Theorem 5.1]{hall1992bootstrap} applies
to \(\frac{1}{\sigma\sqrt n}\sum_{i=1}^n \varphi(Z_i)\), yielding
\begin{align}
	\label{eq:edgeworth_oracle}
	\sup_{t\in\mathbb R}
	\left|
	\mathbb P\left(\frac{1}{\sigma\sqrt n}\sum_{i=1}^n \varphi(Z_i)\le t\right)
	-
	\left(
	\Phi(t)+\frac{1}{\sqrt n}p_1(t)
	\right)
	\right|
	=
	o(n^{-1/2}),
\end{align}
where
\begin{align}
	\label{eq:p1_oracle}
	p_1(t)
	=
	\frac{\kappa_3}{6}(1-t^2)\phi(t),
	\qquad
	\kappa_3
	:=
	\frac{\mathbb E[\varphi(Z)^3]}{\sigma^3}.
\end{align}

Next, let
\begin{align*}
	T_n
	=
	\frac{\sqrt n(\widehat\theta-\theta)}{\widehat\sigma},
	\qquad
	\widehat\sigma^2=a^\prime \widehat V a.
\end{align*}
By Corollary \ref{cl:var}, \(\widehat\sigma^2\to_p \sigma^2\), and therefore
\(\widehat\sigma\to_p \sigma\) and \(\widehat\sigma^{-1}\to_p \sigma^{-1}\). Hence,
\begin{align}
	\label{eq:Tn_slutsky}
	T_n
	=
	\frac{1}{\widehat\sigma}
	\left(
	\frac{1}{\sqrt n}\sum_{i=1}^n \varphi(Z_i)
	\right)
	+o_p(1)
	=
	\frac{1}{\sigma\sqrt n}\sum_{i=1}^n \varphi(Z_i)+o_p(1),
\end{align}
which is sufficient for first-order asymptotic normality.

To obtain an Edgeworth expansion for \(T_n\), we further show the stronger second-order approximation that \(\widehat\sigma\) itself admits a root-\(n\) expansion. 
Specifically, the definition of $\sigma^2$ gives
\begin{align}
	\label{eq:sigma_lin}
	\widehat\sigma^2-\sigma^2
	=
	\frac{1}{n}\sum_{i=1}^n \xi(Z_i)+o_p(n^{-1/2}),~~\text{with}~~
	\mathbb E[\xi(Z)]=0,
\end{align} 
for some measurable \(\xi\) satisfying the moment and Cram\'er-type conditions in
Assumption \ref{as:bootstrap}. 
Then, by a Taylor expansion of the map
\(x\mapsto x^{-1/2}\) around \(x=\sigma^2\),
\begin{align}
	\label{eq:inv_sigma_hat_expansion}
	\widehat\sigma^{-1}
	=
	\sigma^{-1}
	-
	\frac{1}{2\sigma^3}(\widehat\sigma^2-\sigma^2)
	+
	o_p(n^{-1/2}).
\end{align}
Combining \eqref{eq:theta_lin_revised} and \eqref{eq:inv_sigma_hat_expansion}, we obtain
\begin{align}
	\label{eq:Tn_smooth}
	T_n
	=
	G\!\left(
	\frac{1}{n}\sum_{i=1}^n \varphi(Z_i),
	\frac{1}{n}\sum_{i=1}^n \xi(Z_i)
	\right)
	+
	o_p(n^{-1/2})
\end{align}
for a smooth function \(G\) with \(G(0,0)=0\). 
Therefore \(T_n\) is a smooth function of sample averages of i.i.d. variables, and Edgeworth theorem \citep[Equation (5.5) and Theorem 5.1]{hall1992bootstrap} for smooth statistics yields
\begin{align}
	\label{eq:edgeworth_true_revised}
	\sup_{t\in\mathbb R}
	\left|
	\mathbb P(T_n\le t)
	-
	\left(
	\Phi(t)+\frac{1}{\sqrt n}p_{1}(t)
	\right)
	\right|
	=
	o(n^{-1/2}),
\end{align}
where \(p_{1}(t)\) is the first-order correction term associated with the studentized
statistic \(T_n\).

Let $\{Z_1^\ast,...,Z_n^\ast\}$ denote a bootstrap draw. All bootstrap stochastic orders below are understood conditionally on the data, in probability. Since \(\widehat b^\ast\) is the bootstrap estimator, it satisfies the bootstrap first-order condition
\begin{align*}
	\frac{1}{n}\sum_{i=1}^n \psi(Z_i^\ast,\widehat b^\ast)=0.
\end{align*}
Applying a first-order Taylor expansion of $b\mapsto n^{-1}\sum_{i=1}^n \psi(Z_i^\ast,b)$ around \(b=\widehat b\), we obtain
\begin{align*}
	0&=\frac{1}{n}\sum_{i=1}^n \psi(Z_i^\ast,\widehat b^\ast)=\frac{1}{n}\sum_{i=1}^n \psi(Z_i^\ast,\widehat b)+\widehat H(\widehat b^\ast-\widehat b)+R_n^\ast,
\end{align*}
where \(\|R_n^\ast\|=o_{p^\ast}(n^{-1/2})\). 
Subtracting $n^{-1}\sum_{i=1}^n \psi(Z_i,\widehat b)=0$ from both sides yields
\begin{align}
	\label{eq:boot_lin_revised}
	\widehat b^\ast-\widehat b
	=
	\widehat H^{-1}
	\frac{1}{n}\sum_{i=1}^n
	\Bigl(
	\psi(Z_i^\ast,\widehat b)-\psi(Z_i,\widehat b)
	\Bigr)
	+
	r_n^\ast,
\end{align}
where $\|r_n^\ast\|=o_{p^\ast}(n^{-1/2})$.
Since \(\widehat\theta=a^\prime\widehat b\) and \(\widehat\theta^\ast=a^\prime\widehat b^\ast\),
multiplying \eqref{eq:boot_lin_revised} by \(a^\prime\) gives
\begin{align}
	\label{eq:boot_theta_lin}
	\sqrt n(\widehat\theta^\ast-\widehat\theta)
	=
	\frac{1}{\sqrt n}\sum_{i=1}^n
	\Bigl(
	\widehat\varphi(Z_i^\ast)-\widehat\varphi(Z_i)
	\Bigr)
	+
	o_{p^\ast}(n^{-1/2}),
\end{align}
where \(\widehat\varphi(z):=a^\prime \widehat H^{-1}\psi(z,\widehat b)\). Since
\begin{align*}
	\frac{1}{n}\sum_{i=1}^n \widehat\varphi(Z_i)
	=
	a^\prime \widehat H^{-1}\frac{1}{n}\sum_{i=1}^n \psi(Z_i,\widehat b)
	=
	0,
\end{align*}
it follows that \eqref{eq:boot_theta_lin} simplifies to
\begin{align}
	\label{eq:boot_theta_lin_simplified}
	\sqrt n(\widehat\theta^\ast-\widehat\theta)
	=
	\frac{1}{\sqrt n}\sum_{i=1}^n \widehat\varphi(Z_i^\ast)
	+
	o_{p^\ast}(n^{-1/2}).
\end{align}

Next, let \(\widehat\sigma^{\ast2}=a^\prime \widehat V^\ast a\) denote the bootstrap variance estimator and \(\widehat\sigma^2=a^\prime \widehat V a\) its sample analogue. By the bootstrap analogue of \eqref{eq:sigma_lin},
	\begin{align}
		\label{eq:boot_sigma_lin_detailed}
		\widehat\sigma^{\ast2}-\widehat\sigma^2
		=
		\frac{1}{n}\sum_{i=1}^n
		\Bigl(
		\widehat\xi(Z_i^\ast)-\mathbb E^\ast[\widehat\xi(Z_i^\ast)]
		\Bigr)
		+
		o_{p^\ast}(n^{-1/2}),
	\end{align}
	where \(\mathbb E^\ast\) denotes expectation under the bootstrap law conditional on data. Since \(Z_1^\ast,...,Z_n^\ast\) are i.i.d. draws from the empirical distribution conditionally on the data,
	\begin{align*}
		\mathbb E^\ast[\widehat\xi(Z_1^\ast)]
		=
		\frac{1}{n}\sum_{i=1}^n \widehat\xi(Z_i).
	\end{align*}
	Hence \eqref{eq:boot_sigma_lin_detailed} may be written as
	\begin{align}
		\label{eq:boot_sigma_lin_detailed_centered}
		\widehat\sigma^{\ast2}-\widehat\sigma^2
		=
		\frac{1}{n}\sum_{i=1}^n
		\left(
		\widehat\xi(Z_i^\ast)-\frac{1}{n}\sum_{j=1}^n \widehat\xi(Z_j)
		\right)
		+
		o_{p^\ast}(n^{-1/2}).
	\end{align}
	Applying a Taylor expansion of the map \(x\mapsto x^{-1/2}\) around \(x=\widehat\sigma^2\) gives
	\begin{align}
		\label{eq:boot_inv_sigma_detailed}
		(\widehat\sigma^\ast)^{-1}
		=
		\widehat\sigma^{-1}
		-
		\frac{1}{2\widehat\sigma^3}
		\bigl(
		\widehat\sigma^{\ast2}-\widehat\sigma^2
		\bigr)
		+
		o_{p^\ast}(n^{-1/2}).
	\end{align}
	Substituting \eqref{eq:boot_sigma_lin_detailed_centered} into \eqref{eq:boot_inv_sigma_detailed}, we obtain
	\begin{align}
		\label{eq:boot_inv_sigma_detailed_2}
		(\widehat\sigma^\ast)^{-1}
		=
		\widehat\sigma^{-1}
		-
		\frac{1}{2\widehat\sigma^3}
		\left[
		\frac{1}{n}\sum_{i=1}^n
		\left(
		\widehat\xi(Z_i^\ast)-\frac{1}{n}\sum_{j=1}^n \widehat\xi(Z_j)
		\right)
		\right]
		+
		o_{p^\ast}(n^{-1/2}).
	\end{align}	
	Therefore,
	\begin{align*}
		T_n^\ast
		:=
		\frac{\sqrt n(\widehat\theta^\ast-\widehat\theta)}{\widehat\sigma^\ast}
		=
		\sqrt n(\widehat\theta^\ast-\widehat\theta)\,(\widehat\sigma^\ast)^{-1}.
	\end{align*}
	Combining \eqref{eq:boot_theta_lin_simplified} and \eqref{eq:boot_inv_sigma_detailed_2} and expanding the product yields
	\begin{align}
		\label{eq:Tstar_second_order_detailed}
		T_n^\ast
		&=
		\frac{1}{\widehat\sigma}
		\left(
		\frac{1}{\sqrt n}\sum_{i=1}^n \widehat\varphi(Z_i^\ast)
		\right)-
		\frac{1}{2\widehat\sigma^3}
		\left(
		\frac{1}{\sqrt n}\sum_{i=1}^n \widehat\varphi(Z_i^\ast)
		\right)
		\left[
		\frac{1}{n}\sum_{i=1}^n
		\left(
		\widehat\xi(Z_i^\ast)-\frac{1}{n}\sum_{j=1}^n \widehat\xi(Z_j)
		\right)
		\right]
		+
		o_{p^\ast}(n^{-1/2}),
	\end{align}
	where the remainder term is \(o_{p^\ast}(n^{-1/2})\) because
	\(
	n^{-1/2}\sum_{i=1}^n \widehat\varphi(Z_i^\ast)=O_{p^\ast}(1)
	\)
	and
	\(
	n^{-1}\sum_{i=1}^n
	(
	\widehat\xi(Z_i^\ast)-n^{-1}\sum_{j=1}^n \widehat\xi(Z_j)
	)=O_{p^\ast}(n^{-1/2})
	\)
	conditionally on the data. 
	Thus, \(T_n^\ast\) admits an explicit second-order asymptotic expansion, conditionally on the data, as a smooth function of the bootstrap empirical averages generated by \(\widehat\varphi\) and \(\widehat\xi\). This is the bootstrap counterpart of the second-order smooth-statistic representation in \eqref{eq:Tn_smooth}.
	
	Conditionally on the data, \(Z_1^\ast,...,Z_n^\ast\) are i.i.d. draws from the empirical distribution based on \(Z_1,...,Z_n\). In addition, \eqref{eq:Tstar_second_order_detailed} shows that \(T_n^\ast\) admits, conditionally on the data, a second-order smooth-statistic representation of the same functional form as \eqref{eq:Tn_smooth}, with the population moments and cumulants replaced by their empirical analogues. Therefore, under Assumptions \ref{as:bootstrap} \eqref{as:bootstrap1}--\eqref{as:bootstrap3}, the conditional bootstrap version of the smooth-statistic Edgeworth expansion used to obtain \eqref{eq:edgeworth_true_revised} applies, yielding that
		\begin{align}
			\label{eq:edgeworth_boot_revised}
			\sup_{t\in\mathbb R}
			\left|
			\mathbb P^\ast(T_n^\ast\le t)
			-
			\left(
			\Phi(t)+\frac{1}{\sqrt n}\widehat p_{1}(t)
			\right)
			\right|
			=
			o_p(n^{-1/2}),
		\end{align}
		where \(\widehat p_{1}(t)\) is the bootstrap analogue of the first-order correction term \(p_{1}(t)\), obtained by replacing the population moments and cumulants entering \(p_{1}(t)\) with the corresponding sample analogues. In particular, the skewness component is obtained by replacing
		\begin{align*}
			\kappa_3
			=
			\frac{\mathbb E[\varphi(Z)^3]}
			{\bigl(\mathbb E[\varphi(Z)^2]\bigr)^{3/2}}
		\end{align*}
		with
		\begin{align*}
			\widehat\kappa_3
			=
			\frac{\frac{1}{n}\sum_{i=1}^n \widehat\varphi(Z_i)^3}
			{\left(\frac{1}{n}\sum_{i=1}^n \widehat\varphi(Z_i)^2\right)^{3/2}}.
	\end{align*}

By the Weak Law of Large Numbers and Assumptions \ref{as:asymptotic} \eqref{as:asymptotic4}, \ref{as:bootstrap} \eqref{as:bootstrap1}--\eqref{as:bootstrap2}, we obtain 
\begin{align}\label{eq:moment}
\widehat{\kappa}_3 \rightarrow_{p} \kappa_3~~\text{and}~~\sup_t \vert\widehat{p}_1(t) - p_1(t)\vert = O_p(n^{-1/2}).
\end{align}

Combining \eqref{eq:edgeworth_true_revised} and \eqref{eq:edgeworth_boot_revised} yields
\begin{align*}
	\sup_{t \in \mathbb{R}}\left\vert\mathbb{P}^\ast(T_n^\ast \leq t)-\mathbb{P}(T_n \leq t)\right\vert
	=\frac{1}{\sqrt{n}}(\widehat{p}_1(t)-p_1(t))+o(n^{-1/2})=o(n^{-1/2}),
\end{align*}
where the second equality follows by \eqref{eq:moment}.
Thus, the nonparametric bootstrap achieves asymptotic refinement over the normal approximation with $o(n^{-1/2})$. \qed

\color{black}

	\renewcommand\bibname{\LARGE \textbf {Bibliography}}
	\bibliographystyle{chicago}
	\bibliography{IF}	
\end{document}